\documentclass[letterpaper, 12pt, journal, onecolumn]{IEEEtran}

\usepackage{cite}
\usepackage{amsmath,amssymb,amsfonts}
\usepackage{graphicx}
\usepackage{textcomp}
\usepackage{enumerate}
\usepackage[usenames,dvipsnames,table]{xcolor}
\usepackage{soul}
\usepackage[normalem]{ulem}
\usepackage{url}
\usepackage{graphicx}
\usepackage{color}
\usepackage{xcolor}
\usepackage{cite}
\usepackage{amsfonts}
\usepackage{amsthm}
\usepackage{algorithm}
\usepackage{mathtools}
\usepackage{tikz}
\usepackage{pgfplots}
\usepgfplotslibrary{groupplots}
\usepackage{bm}
\usepackage{algorithm}
\usepackage{algpseudocode}
\usepackage{tcolorbox}
\usetikzlibrary{shapes.symbols}
\usetikzlibrary{3d, arrows.meta}
\usepackage{tikz-3dplot}

\newtheorem{theorem}{Theorem}
\newtheorem{proposition}{Proposition}
\newtheorem{assumption}{Assumption}

\newtheorem{lemma}{Lemma}
\newtheorem{remark}{Remark}
\newtheorem{definition}{Definition}

\newcommand\oprocendsymbol{\hbox{$\square$}}
\newcommand\oprocend{\relax\ifmmode\else\unskip\hfill\fi\oprocendsymbol}

\DeclareMathOperator*{\vecc}{vec}

\DeclareMathOperator{\Int}{Int}
\DeclareMathOperator*{\Image}{Im}
\newcommand{\proj}{\operatornamewithlimits{Proj}}

\newcommand{\R}{{\mathbb{R}}}
\newcommand{\Psdef}{{\bm{S}^n_0}}
\newcommand{\Pdef}{{\bm{S}^n_+}}

\newcommand{\hP}{{\hat{P}}}

\newcommand{\hKe}{{\hat{K}_{\textup{a}}}}

\newcommand{\tKe}{{\tilde{K}_\text{a}}}

\newcommand{\dtKe}{{\dot{\tilde{K}}_\text{a}}}

\newcommand{\dhKe}{{\dot{\hat{K}}_{\textup{a}}}}

\newcommand{\mapB}{\mathcal{B}}

\newcommand{\hTe}{{\hat{\theta}_{\text{a}}}}

\newcommand{\tTe}{{\tilde{\theta}_{\text{a}}}}

\newcommand{\dtTe}{{\dot{\tilde{\theta}}_{\text{a}}}}

\newcommand{\dhP}{{\dot{\hat{P}}}}
\newcommand{\hA}{{\hat{A}}}

\newcommand{\dhA}{{\dot{\hat{A}}}}

\newcommand{\tA}{{\tilde{A}}}
\newcommand{\dtA}{{\dot{\tilde{A}}}}

\newcommand{\de}{{\dot{e}}}
\newcommand{\hTa}{{\hat{\theta}_{\text{A}}}}
\newcommand{\dhTa}{{\dot{\hat{\theta}}_{\text{A}}}}

\newcommand{\tTa}{{\tilde{\theta}}_{\text{A}}}

\newcommand{\Lz}{\mathcal{L}}
\newcommand{\Z}{\mathcal{Z}}
\newcommand{\W}{\mathcal{W}}

\newcommand{\A}{{\mathcal{A}}}
\newcommand{\As}{{\mathcal{A}_s}}
\newcommand{\KAL}{\Theta}

\newcommand{\xm}{{x_{\text{m}}}}
\newcommand{\dxm}{{\dot{x}_{\text{m}}}}

\newcommand{\Pta}{{\mathcal{P}(\hA)}}
\newcommand{\Kta}{{K(\hA)}}

\newcommand{\Acl}{{A_{\text{cl}}(\hA)}}

\newcommand{\Ke}{{K_{\textup{a}}(\hA)}}

\newcommand{\dKe}{{\dot{K}_\text{a}}}

\newcommand{\dimx}{{n}} 
\newcommand{\dimu}{{m}}

\title{\Large MR-ARL: Model Reference Adaptive Reinforcement Learning for
	Robustly Stable On-Policy Data-Driven LQR}

\begin{document}

\author{Marco Borghesi, Alessandro Bosso, Giuseppe Notarstefano
	\thanks{M. Borghesi, A. Bosso, and G. Notarstefano are with the Department of Electrical, Electronic, and Information Engineering, Alma Mater Studiorum - Universit\`a di Bologna, 40136 Bologna, Italy. The corresponding author is M.~Borghesi (e-mail: m.borghesi@unibo.it).\\
	A preliminary version of this work is given in \cite{borghesi2023onpolicy}.}
}

\maketitle

\begin{abstract}
    This article introduces a novel framework for data-driven linear quadratic regulator (LQR) design.
    First, we introduce a reinforcement learning paradigm for on-policy data-driven LQR, where exploration and exploitation are simultaneously performed while guaranteeing robust stability of the whole closed-loop system encompassing the plant and the control/learning dynamics.
    Then, we propose Model Reference Adaptive Reinforcement Learning (MR-ARL), a control architecture integrating tools from reinforcement learning and model reference adaptive control.
    The approach stands on a variable reference model containing the currently identified value function.
    Then, an adaptive stabilizer is used to ensure convergence of the applied policy to the optimal one, convergence of the plant to the optimal reference model, and overall robust closed-loop stability.
    The proposed framework provides theoretical robustness certificates against real-world perturbations such as measurement noise, plant nonlinearities, or slowly varying parameters.
    The effectiveness of the proposed architecture is validated via realistic numerical simulations.
\end{abstract}


\section{Introduction}\label{sec:introduction}

	\IEEEPARstart{R}{einforcement Learning} (RL) is a machine learning field that emerged to perform optimization and decision-making by interacting with an environment without or with limited knowledge of its mathematical model \cite{kiumarsi2017optimal, recht2019tour}.
	Over the past years, this field has been successfully applied to multiple domains, including computer games, biology, and economics and finance.
	RL has garnered the attention of the control engineering community, where it has been used to address optimal control in uncertain or model-free scenarios.
	Learning from system data aligns RL with principles found in adaptive control literature \cite{annaswamy2021historical}, which seeks to design dynamic controllers for regulation and tracking in the presence of model uncertainties.
	This work systematically investigates the connection between the fields of optimal and adaptive control, paving the way for a new RL paradigm that provides formal certificates of \emph{robust closed-loop learning and control}, thereby leading to effective performance in real-world applications.
	
	In particular, we focus on solving the infinite-horizon linear quadratic regulator (LQR) problem by developing an \emph{on-policy data-driven algorithm} where data collection and optimization are done simultaneously by applying the learned policy to the actual system.
	The requirements of our framework are schematically presented below in Table I and later formalized in Section \ref{sec:problem_setup}.
	A key distinctive feature of our proposed framework is the requirement of robust asymptotic stability for the whole closed-loop system including both the learning and control dynamics.
	This requirement, as elucidated in the subsequent sections, encapsulates the notion that the proven learning features in nominal cases must persist in perturbed scenarios, encompassing disturbances, measurement noise, slowly varying parameters, and sample-and-hold implementations.
	With a priori guarantees of effective closed-loop controller implementation, our framework is particularly tailored for safety-critical applications, such as collaborative robotics and aircraft control.
	Motivated by the above discussion, we provide an overview of the literature pertaining to data-driven LQR, distinguishing between both so-called off-policy and on-policy approaches.
	Then, we recall model reference adaptive control, one of the main inspirations of the approach of this article.
	
		\begin{tcolorbox}[left=3pt,right=3pt,top=1pt,bottom=1pt, colframe=black, colback=white!10, coltitle=black, colbacktitle=white, title={TABLE I} \textbf{Robustly stable on-policy data-driven LQR}, label=tab:on-policy-dd-lqr]
		{\bfseries Plant:} 
		\begin{equation}\label{eq:plant_dynamics}
			\dot{x} = Ax + Bu,
		\end{equation}
		with state $x\in \R^\dimx$, input $u\in \R^\dimu$, and matrices $A$ and $B$.
		
		{\bfseries Infinite-horizon LQR:} find an optimal control policy $\pi^\star: \mathbb{R}^n \to \mathbb{R}^m$ such that $u(t) = \pi^\star(x(t))$ minimizes, for all initial conditions $x_0 \in \R^{\dimx}$, the following cost functional along the solutions of the plant:
		\begin{equation}\label{prob:cost_functional}
			J(x_0, u(\cdot))\!\coloneqq \!\!\int_0^\infty\!\!\! x(\tau)^\top \!Q x(\tau) + u(\tau)^\top\! R u(\tau) d\tau,
		\end{equation}
		with tuning matrices $Q = Q^\top \geq 0$, $R = R^\top > 0$.
		
		{\bfseries Problem:} with $A$ and $B$ partially or totally unknown, find a controller of the form
		\begin{equation}
			\begin{split}
				\dot{z} & = \varphi(x, z, d) \hspace{1.4cm}\text{\emph{learning dynamics}} \\
				u & = \pi(x, z, d) \hspace{1.4cm}\text{\emph{applied control policy}}
			\end{split}
		\end{equation}
		with $z \in \Z \subset \mathbb{R}^{\ell}$ and $d \in \mathbb{R}^q$ a dither signal, such that the following properties are achieved.
		\begin{enumerate}
			\item {\bfseries\emph{Exploration:}} $d$ probes the uncertainties of $A$ and $B$.
			\item {\bfseries\emph{Exploitation:}} map $x \mapsto \pi(x, z, 0)$ converges to the optimal policy $x \mapsto \pi^\star(x)$.
			\item {\bfseries\emph{Robust stability:}} learning is enforced through robust asymptotic stability of the closed-loop system.
		\end{enumerate}
	\end{tcolorbox}
	
	\subsubsection*{Off-Policy Data-Driven LQR}
	
		Off-policy approaches involve finding the optimal policy without applying it at the same time to the actual system.
		In this context, we find iterative methods, often inspired by the Kleinman algorithm, involving either parameter identification or direct estimate of the policy \cite{watkins1989learning, jiang2012computational, modares2016optimal, pang2018data, krauth2019finite, pang2021robust, lopez2023efficient}.
		Typically, in these methods, the stabilization of the controlled system during the evolution of the learning algorithm is circumvented by assuming an initial stabilizing policy.
		However, as discussed in~\cite{ziemann2022policy}, there are situations where this assumption may be unrealistic due to plant uncertainties. 
		The algorithm \cite{bian2016value} does not need this assumption.
		Finally, other relevant paradigms for off-policy approaches involve batch identification of the policy from data \cite{dorfler2022role, 9992708, de2019formulas, da2018data, de2021low, rotulo2020data} and system-level synthesis \cite{dean2019safely}.
		All these approaches differ from our setup, since the exploration and exploitation phases are not performed simultaneously.
	
	\subsubsection*{On-Policy Data-Driven LQR}
	
		As compared to off-policy approaches, the on-policy paradigm adds the significant challenge of ensuring stability of the interconnection between the plant and the control/learning algorithm.
		Early attempts to address this setup are \cite{bradtke1994adaptive,fazel2018global,vrabie2009adaptive,modares2014linear,rizvi2017output,rizvi2019reinforcement,kiumarsi2015optimal}, where the stabilizing feedback gain is updated at discrete iterations.
		However, the stability of the whole closed-loop system is not analyzed and an initial stabilizing policy is required, similar to off-policy approaches.
		A data-driven approach to compute the initial gain is presented in \cite{possieri2022q}.
		To the best of our knowledge, \cite{possieri2022value} is the only work in the literature that provides stability guarantees without a stabilizing initialization, although the focus is on the learning dynamics and not the overall closed-loop system.
	
	\subsubsection*{Model Reference Adaptive Control}
	
		We finally review the literature dealing with model reference adaptive control (MRAC).
		The principle of this technique is to match the unknown system dynamics to a reference model with desired properties \cite{tao2014multivariable, ioannou1996robust, narendra2012stable}.
		To ensure design feasibility, this stabilization technique requires the plant to satisfy constraints named \emph{matching conditions}.
		A recent work combining MRAC and RL is \cite{guha2021online}, where RL techniques are used to find the optimal controller for a reference model based on nominal plant parameters.
		Then, MRAC is applied to assign the reference model to the real system.
		However, convergence to the desired policy is not proved and can only be ensured to a policy that is optimal for the reference model and not the true system.
	
	\subsubsection*{Article Contribution}
	
		The goal of this article is to lay the foundations for a new paradigm of on-policy data-driven LQR according to the problem setting described in Table~I.
		The main paper contribution is twofold: {\textbf{(i)}} introducing a novel formulation of the on-policy data-driven LQR problem where centrality is given to the stability of the whole closed-loop learning and control system; {\textbf{(ii)}} providing a combined adaptive control and reinforcement learning design paradigm to address this framework.
		
		Concerning our first contribution, we formulate the on-policy data-driven LQR problem in terms of convergence of the controller, the plant, and an exosystem (modeling the dither signal) to an asymptotically stable set.
		The fundamental property defining this set is that the learned policy is optimal.
		Additionally, the set becomes smaller as the dither amplitude is reduced.
		Thanks to this formulation, we ensure that asymptotic stability in the nominal scenario is preserved, practically and semiglobally, also for a broad class of perturbations, see \cite[Ch. 7]{goebel2012hybrid}.
		With the generality of the proposed framework, we aim to provide a solid foundation for future work in the field.
		
		The second and main contribution of this work consists in introducing an on-policy learning and control law, termed \emph{Model Reference Adaptive Reinforcement Learning} (MR-ARL), integrating concepts from system identification, adaptive control, and reinforcement learning paradigms.
		The architecture is structured as a modular actor-critic system with a time-varying reference model bridging the two modules.
		The actor, inspired to a MRAC architecture, guides the plant to a desired behavior set by the reference model, even in the presence of parametric uncertainties.
		The reference model is updated online by the critic, which leverages system identification techniques to estimate the dynamics.
		To impose the desired learning properties, the reference model is driven by a dither signal for which we require suitable richness properties.
		By relying on analysis tools related to adaptive control, differential inclusions, and singular perturbations, we prove formally that our architecture achieves the following properties for the whole closed-loop system: (i) convergence of the policy to the optimal one; (ii) asymptotic estimation of the true system parameters; (iii) uniform asymptotic stability of an attractor (tunable with the dither amplitude); (iv) robustness in the sense of semiglobal practical asymptotic stability with respect to unmodeled nonlinearities and disturbances.
		To the best of authors' knowledge, in the context of on-policy data-driven LQR, this algorithm is the first one possessing all these properties.
		Also, no assumptions are needed about the initial policy.
		Further, persistency of excitation, needed to ensure convergence, is not assumed a priori, but rather guaranteed by design by resorting to concepts from nonlinear adaptive systems~\cite{panteley2001relaxed}.
		Finally, given the inherent robustness of the proposed design framework, we ensure that the algorithm is effective in the presence of slowly varying parameters.
		To validate this property, our numerical simulations cover both the constant parameters case and the one with drifts.
		
		In Section \ref{sec:problem_setup}, we provide preliminary concepts of LQR and introduce the on-policy data-driven paradigm that formalizes Table I.
		Then, Section \ref{sec:controller} presents MR-ARL and its derivation, while Section \ref{sec:main_result} provides its stability properties.
		The technical results for the stability analysis are given in Section \ref{sec:analysis}, while all related proofs are left in the Appendix.
		Finally, Section \ref{sec:example} provides in-depth numerical results.
	
	\subsubsection*{Notation}
		The set of real numbers is denoted by $\mathbb{R}$, while
		$\mathbb{R}_{\geq 0} \coloneqq [0, +\infty)$.  $(\cdot)^\top$
		and $(\cdot)^\dagger$ denote the transpose and the Moore-Penrose
		pseudo-inverse of real matrices, while $\Psdef$ (resp. $\Pdef$) denotes
		the set of $n \times n$ symmetric, positive semidefinite
		(resp. positive definite) real matrices. 
		$|\cdot|$ denotes the induced $2-$norm of real matrices, while $|\cdot|_F$ indicates the Frobenius norm.
		The notation $\dot{\xi} = f(t, \xi)$, $\xi \in C\subset\mathbb{R}^r$ represents a differential equation having flow set $C$, i.e., with initial state
		and flow constrained to $C$. We refer to
		\cite{goebel2012hybrid} for the solution and stability notions of
		(hybrid) dynamical systems, covering the constrained differential
		equations of this article.

\section{Preliminaries and Problem Setup}\label{sec:problem_setup}

	Following the discussion in the introduction, we now provide a rigorous formulation of the on-policy data-driven linear quadratic regulation (LQR) problem addressed in the paper.
	
	\subsection{Linear Quadratic Regulation}
	
		We start by introducing the basic concepts of LQR for system \eqref{eq:plant_dynamics} and the cost functional \eqref{prob:cost_functional}.
		The infinite-horizon LQR problem involves finding a \emph{control policy} $\pi^\star: \R^{\dimx} \to \R^{\dimu}$ such that applying $u(t) = \pi^\star(x(t))$ for all $t \in \mathbb{R}_{\geq 0}$ solves, for all initial conditions $x_0 \in \R^{\dimx}$, the following optimal control problem:
		\begin{equation}\label{eq:LQR}
			\begin{split}
				\min_{u(\cdot)} J(x_0, u(\cdot))&\coloneqq \int_0^\infty\!\! x(\tau)^\top\! Q x(\tau) + u(\tau)^\top\! R u(\tau) d\tau\\
				\text{subject to:} \;\;\;& \dot{x}(t) = Ax(t) + Bu(t), \;\;\; \forall t \in \mathbb{R}_{\geq 0},\\
				& x(0)=x_0.
			\end{split}
		\end{equation}
		Under the assumption that pair $(A, B)$ be stabilizable and $(\sqrt{Q}, A)$ be detectable, the LQR problem \eqref{eq:LQR} is solved by the linear policy:
		\begin{equation}\label{eq:optimal_policy}
			\pi^\star(x) \coloneqq K^\star x, \qquad K^\star \coloneqq -R^{-1}B^\top P^\star,
		\end{equation}
		where $P^\star \in \Psdef$ is the unique solution of the algebraic Riccati equation (ARE):
		\begin{equation}\label{eq:ARE}
			A^\top P^\star + P^\star A - P^\star BR^{-1}B^\top P^\star + Q = 0.
		\end{equation}
		Additionally, $P^\star \in \Pdef$ if pair $(\sqrt{Q}, A)$ is observable.
		We also recall that $P^\star$ specifies the \emph{value function}, which is defined, for a given initial condition $x$, as the minimum of $J(x, u(\cdot))$:
		\begin{equation}\label{eq:optimal_value_function}
			V^\star(x) \coloneqq \min_{u(\cdot)}J(x, u(\cdot)) \coloneqq x^\top P^\star x.
		\end{equation}
		Consider the differential Riccati equation (DRE) with flow set constrained to symmetric positive semidefinite matrices:
		\begin{equation}\label{eq:DRE}
			\dot{P} = A^\top P + P A - P BR^{-1}B^\top P + Q, \quad P \in \Psdef.
		\end{equation}
		Then, if $(A, B)$ is stabilizable and $(\sqrt{Q}, A)$ is detectable, the equilibrium point $P^\star$ is uniformly globally asymptotically stable (UGAS) for the constrained differential equation \eqref{eq:DRE} \cite[Thm. 15]{kuvcera1973review}.
		Furthermore, if $(A, B)$ is controllable and $(\sqrt{Q}, A)$ is observable, $P^\star$ is uniformly locally exponentially stable (ULES) \cite[Thm. 4]{bucy1967global}.
		Formal results describing the stability properties of DRE \eqref{eq:DRE} are provided in \cite{kuvcera1973review, bucy1967global}.
	
	\subsection{Robustly Stable On-Policy Data-Driven LQR}
	
		Solving the LQR problem \eqref{eq:LQR} involves computing the solution $P^\star$ of ARE \eqref{eq:optimal_policy}, which depends on the matrices $A$ and $B$ of plant \eqref{eq:plant_dynamics}.
		Therefore, if $A$ and $B$ are unknown or only partially known, it is necessary to resort to data-driven approaches based on acquiring the data of the state-input sequences $(x(t), u(t))$ (continuously or via batches).
		
		In this work, we are interested in finding an on-policy data-driven algorithm, where data collection and learning are performed simultaneously by applying the learned policy.
		We now provide a novel rigorous framework to formalize this problem so that its solution guarantees, by design, desirable learning and robust stability properties.
		
		As anticipated in the introduction, the class of controllers that we seek are described by continuous-time dynamical systems of the form
		\begin{equation}\label{eq:controller_generic}
			\begin{split}
				\dot{z} &= \varphi(x, z, d)\\
				u &= \pi(x, z, d)\\
			\end{split}\qquad z \in \Z,
		\end{equation}
		where $z$ is the controller state, $\Z \subset \mathbb{R}^{\ell}$ is a closed set, $d \in \mathbb{R}^{q}$ is a uniformly bounded signal, named hereafter \emph{dither}, while $\varphi$ and $\pi$ are maps that are locally Lipschitz in their arguments.
		To the algorithm \eqref{eq:controller_generic}, we associate the \emph{learning set} $\Lz$, defined as:
		\begin{equation}\label{eq:learning_set}
			\Lz \coloneqq \{z \in \Z: \pi(x, z, 0) = K^\star x,\; \forall x \in \mathbb{R}^n\},
		\end{equation}
		which denotes the set of all controller states such that the control policy $\pi$ coincides with the optimal policy $\pi^\star$ in \eqref{eq:optimal_policy} whenever the dither $d$ is turned off.
		
		In \eqref{eq:controller_generic}, $d$ is an exogenous signal that may include references for tracking, probing signals, or disturbances.
		To ensure well-posedness of the problem formulation, from now on we consider a general class of signals $d$ that can be generated by an autonomous dynamical system (\emph{exosystem}) of the form
		\begin{equation}\label{eq:exosystem}
			\begin{split}
				\dot{w} &= s(w)\\
				d &= h(w)
			\end{split}\qquad w \in \W,
		\end{equation}
		where $w$ is the exosystem state, $\W \subset \mathbb{R}^{p}$ is the set of admissible initial conditions $w(0)$, while $s$ and $h$ are locally Lipschitz maps.
		Since $d$ is a bounded signal defined for all $t \in \mathbb{R}_{\geq
			0}$, we suppose that the set $\W$ be compact and strongly forward invariant under the flow of \eqref{eq:exosystem}.
		\begin{remark}
			Exosystem \eqref{eq:exosystem} is not implemented in the control solution but is used to represent the class of admissible signals $d$.
			Moreover, the results of the paper hold if \eqref{eq:exosystem} is replaced by a well-posed hybrid dynamical system \cite{goebel2012hybrid} to include discontinuous dither signals.
			Here, we use a continuous-time exosystem to avoid an additional notational burden.
		\end{remark}
		
		The closed-loop system resulting from the interconnection of exosystem \eqref{eq:exosystem}, plant \eqref{eq:plant_dynamics}, and controller \eqref{eq:controller_generic} is given by
		\begin{equation}\label{eq:closed-loop_system}
			\begin{split}
				\dot{w} &= s(w)\\
				\dot{x} &= Ax + B\pi(x, z, h(w))\\
				\dot{z} &= \varphi(x, z, h(w))\\
			\end{split}\qquad (w, x, z) \in \W\times\mathbb{R}^n\times\Z.
		\end{equation}
		We are ready to precisely state the requirements for controller \eqref{eq:controller_generic}, which include a precise stability characterization for the closed-loop system \eqref{eq:closed-loop_system}.
		
		\begin{definition}\label{def:design_solution}
			We say that controller \eqref{eq:controller_generic} \emph{solves the robustly stable on-policy data-driven LQR problem} if the learning set $\Lz$ in \eqref{eq:learning_set} is non-empty and, for a chosen class of dither signals $d$ generated by exosystem \eqref{eq:exosystem}, there exists a compact attractor $\mathcal{A}$, satisfying
			\begin{equation}
				\mathcal{A} \subset \W\times\mathbb{R}^n\times\Lz,
			\end{equation}
			that is asymptotically stable for the closed-loop system \eqref{eq:closed-loop_system}.
		\end{definition}
		
		We show that any algorithm satisfying Definition \ref{def:design_solution} covers all of the design requirements stated in the introduction.
		\begin{itemize}
			\item {\textbf{Exploration}:} choosing the dither $d$ determines the shape and the attractivity properties of $\mathcal{A}$, thus it ensures the necessary probing to estimate the optimal policy.
			\item {\textbf{Exploitation:}} since the projection of $\mathcal{A}$ in the $z$ direction is a subset of the learning set $\Lz$, uniform attractivity of $\mathcal{A}$ (encoded in asymptotic stability) ensures $z \to \Lz$ and, thus, correct estimation of the optimal policy.
			\item {\textbf{Robust stability:}} under the regularity properties required for the controller and assumed for the exosystem, asymptotic stability of the attractor $\mathcal{A}$ is preserved (practically and semiglobally) under a broad range of non-vanishing perturbations arising in real-world scenarios, such as disturbances, measurement noise, unmodeled dynamics, sample-and-hold implementations of the controller, and actuator dynamics.
		\end{itemize}
		\begin{remark}
			In Definition \ref{def:design_solution}, we do not specify the restrictions on the knowledge of $A$ and $B$ to cover a broad range of applications and solutions.
			However, the prior knowledge on the parametric uncertainties determines the design of $\varphi$, $\pi$, and $\Z$.
			Note that, if controller \eqref{eq:controller_generic} is not appropriately chosen, the learning set may be empty.
		\end{remark}
		\begin{remark}
			The convergence of $z$ to the learning set $\Lz$ in Definition \ref{def:design_solution} implies that the controlled plant becomes asymptotically:
			\begin{equation}\label{eq:optimal_closed_loop}
				\dot{x} = (A + BK^\star)x + \Delta(x, z, d),
			\end{equation}
			where $\Delta(x, z, d) \coloneqq B(\pi(x, z, d) - K^\star x)$ vanishes in $d = 0$.
			Moreover, since $\Delta(x, z, d)$ is locally Lipschitz, it is uniformly bounded for all $(x, z, w)$ in the compact attractor $\mathcal{A}$.
			As a consequence, the input-to-state stability of \eqref{eq:optimal_closed_loop} implies that
			\begin{equation}
				\limsup_{t\to \infty} |x(t)| \leq \alpha(\limsup_{t\to\infty}|d(t)|),
			\end{equation}
			where $\alpha$ is a class $\mathcal{K}$ function.
			In other words, the ultimate bound of $x$ directly depends on the amplitude of the injected dither $d$.
		\end{remark}

\section{Model Reference Adaptive Reinforcement Learning}\label{sec:controller}
	\begin{figure}[t!]
		\centering
		\usetikzlibrary{shapes.symbols}

\begin{tikzpicture}[scale = 1.0]

	
	

	\begin{scope}[xshift=0cm, yshift = 0cm]
		
		
		\filldraw[fill=gray!1, rounded corners=1mm](-0.2,-0.7) rectangle (9.2,3.7);
		\node[anchor = center] at (4.3, -0.35) {\small Model Reference Adaptive \hspace{0.1cm}Reinforcement Learning};
		
		\filldraw[fill=gray!1, rounded corners=1mm](0,0) rectangle (4,3.5);
		\node[anchor = center] at (0.8, 0.4) {\small ACTOR};
		
		\filldraw[fill=gray!1, rounded corners=1mm](5,0) rectangle (9, 3.5);
		\node[anchor = center] at (8.2, 0.4) {\small CRITIC};
		
		\filldraw[fill=gray!10, rounded corners=1mm](2,4.) rectangle (7,5.);
		\node[anchor = center] at (4.5, 4.7) {\small Plant};
		\node[anchor = center] at (4.5, 4.3) {$x$};
		
		\filldraw[fill=gray!10, rounded corners=1mm](2.5, 0.2) rectangle (6.5,1.5);
		\node[anchor = center] at (4.5, 1.2) {\small Reference model};
		\node[anchor = center] at (4.5, 0.6) {\small$\xm$};
		
		\filldraw[fill=gray!10, rounded corners=1mm](0.1,1.7) rectangle (3.9,3.3);
		\node[anchor = center] at (2, 2.9) {\small Adaptive stabilizer};
		\node[anchor = center] at (2, 2.3) {\small $\hKe$};

		\filldraw[fill=gray!10, rounded corners=1mm](5.1,1.7) rectangle (8.9,3.3);
		\node[anchor = center] at (7, 2.9) {\small Value\hspace{-1pt} function\hspace{-1pt} identifier};
		\node[anchor = center] at (7, 2.3) {\small $\xi, \zeta, \hA, \hP$};
		
		
		\node[anchor = west] at (4.5, 2.8) {\footnotesize{$x$}};
		\draw[->, thick] (4.5, 4.) -- (4.5, 2.6) -- (4, 2.6);
		\draw[->, thick] (4.5, 4.) -- (4.5, 2.6) -- (5, 2.6);
		
		\node[anchor = west] at (1., 1.2) {\footnotesize{$\xm, d$}};
		\draw[->, thick] (2.5, 0.8) -- (1, 0.8) -- (1, 1.7);
		
		\node[anchor = west] at (8, 1.2) {\footnotesize{$\hA, \hP$}};
		\draw[->, thick] (8, 1.7) -- (8, 0.8) -- (6.5, 0.8);
		
		\node[anchor = south west] at (1., 4.5) {\footnotesize{$u$}};
		\draw[->, thick] (1, 3.5) -- (1 , 4.5) -- (2, 4.5);

		\draw[<->, thick, dashed]  (4 , 2.) -- (5, 2.);

		\node[cloud,
		draw =black,
		text=black,
		fill = gray!10,
		aspect=2] (c) at (4.5,-2.2){\small Exosystem};
		\node at (4.5,-2.6) {\small{$w$}};
		\draw[->, thick, dashed]  (4.5 , -1.2) -- (4.5, 0.2);
		\node[anchor = west] at (4.5, -1) {\small $d$};

	\end{scope}

\end{tikzpicture}
		\caption{Block scheme of the \emph{Model Reference Adaptive Reinforcement Learning}.}
		\label{fig:block_scheme} 
	\end{figure}
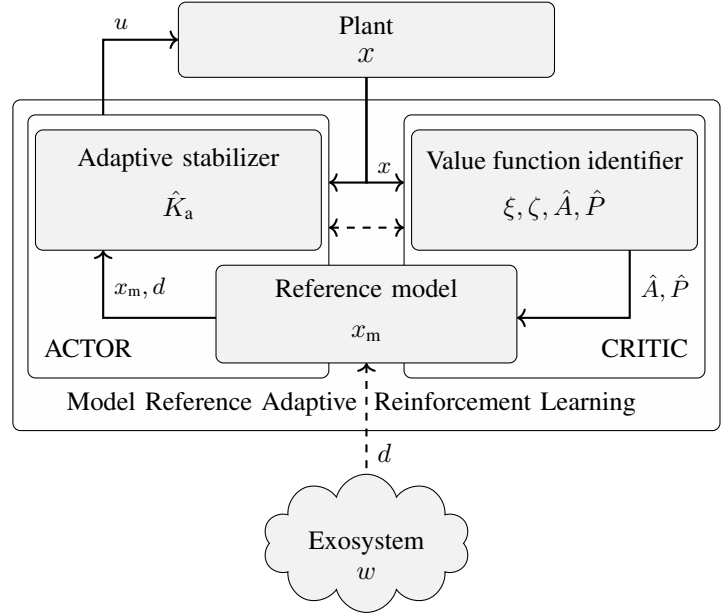
	We now present a new control and learning approach, named \emph{Model Reference Adaptive Reinforcement Learning} (MR-ARL), which satisfies Definition \ref{def:design_solution} in a scenario of \emph{structured
		uncertainties} characterized by the following assumptions.
	\begin{assumption}\label{as:ctrl_det}
		There exists a known closed convex set $\mathcal{C} \subset \R^{\dimx\times\dimx}$ such that:
		\begin{enumerate}
			\item $\mathcal{C}$ has a non-empty interior and $A \in \Int(\mathcal{C})$;
			\item $(\hA, B)$ is controllable and $(\sqrt{Q}, \hA)$ is observable for all $\hA \in \mathcal{C}$.
		\end{enumerate}
	\end{assumption}
	\begin{remark}
		From \cite{possieri2022value, menini2020algebraic}, it is known that if there exists $A_0$ such that $(A_0, B)$ is controllable and $(\sqrt{Q}, A_0)$ is observable, then there exists a scalar $\rho > 0$ such that $(A, B)$ is controllable and $(\sqrt{Q}, A)$ is observable for all $A$ such that $|A - A_0|\leq \rho$.
		Therefore, during implementation, the set $\mathcal{C}$ can be chosen as a ball centered in a nominal value $A_0$ of $A$ with radius $\rho$ chosen to include all possible uncertainties.
	\end{remark}
	\begin{assumption}\label{as:matching}
		Consider the linear map $\mapB:K \in \R^{\dimu\times \dimx} \longmapsto BK \in \R^{\dimx \times \dimx}$,
		where $B$ is the input matrix in \eqref{eq:plant_dynamics}.
		For some known $A_0\in \mathcal{C}$, it holds that
		\begin{equation}\label{eq:new_matching_condition_existence}
			A_0 - A \in \Image (\mapB). 
		\end{equation}
	\end{assumption}
	\begin{remark}\label{rem:M}
		Assumption \ref{as:matching} is an alternative formulation of the matching conditions used in the MRAC literature \cite{tao2014multivariable}. 
		Notice that for any $\hA$ such that $\hA-A_0 \in \Image(\mapB)$, by $\hA - A_0 +(A_0 - A) = \hA - A \in \Image(\mapB)$, there exists $K_{\textup{a}} \in \R^{\dimu\times \dimx}$ such that 
		\begin{equation}\label{eq:matching_condition_existence}
			\hA-A = B K_{\textup{a}}.
		\end{equation}
		Given the two sets $\mathcal{C}$ and $A_0+\Image(\mapB)$, we are interested in all matrices $\hA \in \KAL$, where 
		\begin{equation}\label{eq:Theta_definition}
			\KAL \coloneqq \mathcal{C}\cap(A_0+\Image(\mapB))
		\end{equation}
		since they can be used to build reference models. 
	\end{remark}
	
	Our controller is conceived as an \emph{actor-critic} modular architecture where a \emph{reference model} bridges the two parts of the design.
	The resulting structure is a MRAC where the reference model is continuously updated with value iteration, thus we aptly name it \emph{Model Reference Adaptive Reinforcement Learning}.
	We introduce the building blocks of the design.
	
	\begin{itemize}
		\item {\textbf{Critic:}} this block performs \emph{data-driven value function identification} to build an optimal and asymptotically stable reference model.
		In particular, a gradient identifier computes an estimate $\hA \in \KAL$ of $A$ that is used to obtain an estimate $\hP$ of the solution $P^\star$ of ARE~\eqref{eq:ARE}.
		In this respect, Assumption \ref{as:ctrl_det} guarantees that for any estimate $\hA$ the computation of $\hP$ is feasible.
		Then, $\hA$ and the optimal gain estimate $-R^{-1} B^{\top} \hP$ are used to build a reference model having state matrix $\hA - BR^{-1} B^{\top} \hP$.
		As input to the reference model, we consider a dither $d$ with sufficient richness properties to ensure convergence to the true system parameters and to the optimal policy.
		\item {\textbf{Actor:}} this block assigns the input to the plant to \emph{adaptively track the reference model}.
		During the transient, the feedback gain $-R^{-1} B^{\top} \hP$ may be not stabilizing for the real system. 
		For this reason, the actor introduces in the control law an additional adaptive feedback gain $\hKe$ to cancel the mismatch between the estimated matrix $\hA$ and the real $A$.
		Canceling such a mismatch is possible due to Assumption \ref{as:matching}.
	\end{itemize}
	
	See Fig.~\ref{fig:block_scheme} for a block scheme of Model Reference Adaptive  Reinforcement Learning.
	The full description of the design is presented in Algorithm~\ref{alg:MRARL} and discussed in detail in the next subsections.
	
	\begin{algorithm}[t!]
		\caption{MR-ARL}\label{alg:MRARL}
		\begin{algorithmic}
			\State \hspace{-0.47cm} \textbf{Initialization:} 
			\State $\hP(0) \in \Psdef$, $\hA(0) \in \KAL$, with $\KAL$ from Remark \ref{rem:M}
			\State $\lambda, \gamma, \nu, g, \mu > 0$ design gains
			\State $d(t)$: bounded and stationary signal with each entry sufficiently rich of order $n+1$ and uncorrelated
			\State \hspace{-0.47cm} \textbf{Repeat:} 
			\State \emph{Swapping filters:} 
			\begin{equation}\label{eq:swapping}
				\begin{split}
					\dot{\xi} &= - \lambda \xi + x, \qquad \dot{\zeta} = - \lambda(x + \zeta) - Bu\\
				\end{split}
			\end{equation}
			\State \emph{Identifier dynamics:}
			\begin{equation}\label{eq:identifier}
				\begin{split}
					&\dhA = \proj_{\hA \in \mathcal{C}}\!
					\left\{\!- \gamma BB^\dagger \frac{\epsilon \xi^\top}{1 + \nu |\xi||\epsilon|}\!\right\}\!\!, \quad \epsilon \coloneqq \hA \xi \! - \! (x + \zeta)
				\end{split}
			\end{equation}
			\State \emph{Value iteration:}
			\begin{equation}\label{eq:DRE_implementation}
				\dhP \!= g \!\left(\!\hA^\top\! \hP + \hP\hA - \hP B R^{-1}\! B^\top\! \hP + Q\right)
			\end{equation}
			\State \emph{Reference model:}
			\begin{equation}\label{eq:model_reference}
				\dxm = (\hA - BR^{-1} B^{\top} \hP) \xm + B d
			\end{equation}
			\State \emph{Adaptive gain dynamics:} 
			\begin{equation}\label{eq:adaptive_dynamics}
				\dhKe = -\mu B^\top \hP (x-\xm) x^\top + B^\dagger \dhA
			\end{equation}
			\State \emph{System input:}
			\begin{equation}\label{eq:system_input}
				u= - R^{-1} B^{\top} \hP x + \hKe x +d
			\end{equation}
		\end{algorithmic}
	\end{algorithm}
	
	\subsection{Critic: Value Function Identifier}
		In this subsection, we build a continuous-time identifier of $P^\star$ based on the estimation of matrix $A$. 
		Given the structure of system \eqref{eq:plant_dynamics}, we compute an estimate $\hA \in \KAL$ of $A$ by designing a swapping filter of the form~\eqref{eq:swapping}, with $\lambda > 0$ a scalar gain for tuning the filter time constant.
		Using the filter states, define the \emph{prediction error}
		\begin{equation}\label{eq:prediction_error}
			\epsilon \coloneqq \hA \xi- (x + \zeta),
		\end{equation}
		which we can rewrite as $\epsilon = (\hA-A)\xi + \tilde{\epsilon}$, where
		\begin{equation}\label{eq:tilde_eps}
			\tilde{\epsilon} \coloneqq A\xi - (x + \zeta),
		\end{equation}
		is an error signal that is shown in Section \ref{sec:analysis} to converge exponentially to zero.
		Since the available signal $\epsilon$ converges exponentially to $(\hA - A)\xi$, which contains the parameter estimation error $\hA - A$, we can use the normalized projected gradient descent algorithm \eqref{eq:identifier} to update the estimate $\hA$.
		In \eqref{eq:identifier}, parameters $\gamma > 0$ and $\nu > 0$ are scalar gains, while
		the multiplicative term $BB^\dagger$ is a projection onto $\Image(\mapB)$ \cite[Sec. 5.5.4]{golub2013matrix}.
		This projection is needed to ensure, given any initialization $\hat{A}(0)\in A_0+\Image(\mapB)$, that the estimate $\hA(t)$ never leaves this subspace. 
		Finally, $\proj_{\hA \in \mathcal{C}}\{\cdot\}$ is a Lipschitz continuous parameter projection operator, whose expression is provided, e.g., in \cite[Appendix E]{krstic1995nonlinear} and depends on the shape of the set $\mathcal{C}$.
		
		Given the estimate $\hA$, we are interested in computing the matrix $\hP \in \Pdef$ that solves the ARE
		\begin{equation}\label{eq:ARE_implementation}
			\mathcal{R}(P, \hA) \coloneqq \hA^\top P + P\hA - P B R^{-1} B^\top P + Q=0.
		\end{equation}
		From \cite{borghesi2023onpolicy}, such a matrix could be obtained by computing the map $\mathcal{P}(\hA)$ that solves \eqref{eq:ARE_implementation} for each $\hA$, i.e., such that:
		\begin{equation}\label{eq:ARE_static}
			\mathcal{R}(\mathcal{P}(\hA), \hA) = 0, \quad \text{for all } \hA \in \KAL.
		\end{equation}
		For simplicity in the implementation and inspired by \cite{possieri2022value}, in Algorithm~\ref{alg:MRARL}, we compute $\hP$ via the dynamical system \eqref{eq:DRE_implementation}, which is a DRE rescaled by the tuning gain $g > 0$.
		Notice that, if $\hA$ is constant, then the solution $\hP$ of \eqref{eq:DRE_implementation} converges to $\mathcal{P}(\hA)$.
		
		\begin{remark}
			Assumption \ref{as:ctrl_det} guarantees that, for each $\hA \in \KAL$, $\mathcal{P}(\hA)$ exists, is unique, and positive definite for any $\hA \in \KAL$.
			Although stabilizability of $(\hA, B)$ would be sufficient in Assumption \ref{as:ctrl_det} for the solvability of $\mathcal{R}(P, \hA) = 0$, controllability is essential to guarantee convergence of the identifier under sufficient richness of the dither $d$.
			Moreover, we need observability instead of simple detectability to ensure that $\mathcal{P}(\hA)$ is positive definite. 
		\end{remark}
		\begin{remark}\label{remark:cl_stability}
			From Assumption \ref{as:ctrl_det} and the parameter projection in \eqref{eq:identifier}, matrix $\hA - BR^{-1} B^{\top} \mathcal{P}(\hA)$ is Hurwitz by design.
			Therefore, if $\hA$ converges to a constant matrix, \eqref{eq:DRE_implementation} ensures that $\hA - BR^{-1} B^{\top} \hP$ converges to a Hurwitz matrix.
		\end{remark}
	
	\subsection{Reference Model}
		Given the estimate $\hP$ of $P^\star$, we design a reference model for system \eqref{eq:plant_dynamics}.
		The reference model has to embed all the properties required for the plant, i.e., robust stability, optimality and persistency of excitation.
		To these aims, consider system \eqref{eq:model_reference}, where $\xm \in \R^\dimx$ is the reference model state.
		We embed the stability and optimality properties through $\hA - BR^{-1} B^{\top} \hP$, which is designed to converge to a Hurwitz matrix.
		\begin{remark}
			Different from classic MRAC, the state matrix $\hA - BR^{-1} B^{\top} \hP$ of the reference model \eqref{eq:model_reference} is not constant but time-varying as it depends on the estimates $\hA$ and $\hP$.
			This property leads to an adaptive design where the known-plant stabilizing gains are time-varying.
		\end{remark}
		Finally, we embed the persistency of excitation properties through dither $d \in \R^\dimu$, which is chosen such that it is a bounded stationary signal, whose entries are sufficiently rich of order $n+1$ and uncorrelated. 
		\begin{remark}
			We model $d(t)$ as the output of an exosystem of the form \eqref{eq:exosystem}.
			It is not necessary to actually implement the exosystem as part of the algorithm, as we show in the numerical example. 
		\end{remark}
	
	\subsection{Actor: Model Reference Adaptive Controller}
		Given the reference model \eqref{eq:model_reference}, we design an adaptive controller for system \eqref{eq:plant_dynamics}.
		Define the tracking error $e\coloneqq x - \xm$ and compute its time derivative from \eqref{eq:plant_dynamics}, \eqref{eq:model_reference} as
		\begin{equation}\label{eq:error_dynamics}
			\begin{split}
				\dot{e} =&\; Ax + B u - (\hA - BR^{-1} B^{\top} \hP)(x - e) - B d\\
				=&\; (\hA - BR^{-1} B^{\top} \hP) e + (A- \hA)x + B(u + R^{-1} B^{\top} \hP x - d).
			\end{split}
		\end{equation}
		To ensure that the plant \eqref{eq:plant_dynamics} asymptotically copies the behavior of the reference model \eqref{eq:model_reference}, i.e., $e(t) \to 0$, we exploit the fact that $\hA \in \KAL$.
		In particular, from \eqref{eq:matching_condition_existence}, for each $\hA$ there exists $\Ke\in \R^{\dimu\times\dimx}$ such that
		\begin{equation}\label{eq:matching_condition}
			\hA - A = B \Ke. 
		\end{equation}
		More specifically, we can ensure that map $\Ke$ is smooth in $\hA$ by choosing:
		\begin{equation}\label{eq:explicit_Ke}
			\Ke \coloneqq  B^\dagger (\hA-A),
		\end{equation}
		where $B^\dagger$ denotes the Moore-Penrose pseudoinverse of $B$.
		This way, \eqref{eq:error_dynamics} becomes
		\begin{equation}\label{eq:error_dynamics_2}
			\begin{split}
				\dot{e} =&\; (\hA - BR^{-1} B^{\top} \hP) e + B(u + R^{-1} B^{\top} \hP x -\Ke x - d),
			\end{split}
		\end{equation}
		suggesting a control law of the form 
		\begin{equation}
			\begin{split}
				u &\coloneqq -R^{-1} B^{\top} \hP x + \Ke x +d
			\end{split}
		\end{equation}
		if the plant dynamics were known.
		However, $\Ke$ is unavailable for design as it depends also on $A$, as highlighted in \eqref{eq:matching_condition}, thus we consider the certainty-equivalence-based adaptive controller given in~\eqref{eq:system_input}, where $\Ke$ is replaced by the adaptive gain $\hKe$, driven by the adaptive law~\eqref{eq:adaptive_dynamics} where $\mu > 0$ is a scalar gain. 
		The first term in the adaptive law~\eqref{eq:adaptive_dynamics} is a standard update to ensure the error $e$ goes asymptotically to zero in a framework where the model mismatch is constant.
		However, since $\hA$ is continuously updated by identifier~\eqref{eq:identifier}, the second term in the update law takes into account the time-varying mismatch.
	
	\section{Main Result}\label{sec:main_result}
		We now provide the main results of this work, where we show that Model Reference Adaptive Reinforcement Learning solves the robustly stable on-policy data-driven LQR problem as per Definition \ref{def:design_solution}.
		The first result is given supposing to have a DRE dynamics in \eqref{eq:DRE_implementation} infinitely faster than the rest of the system (\emph{reduced-order system}), i.e., supposing $\hP(t)=\mathcal{P}(A(t))$ as in \eqref{eq:ARE_static} at each $t$.
		For this reason, we mark the results for the reduced-order system with a subscript $s$ to highlight its slow dynamics.
		We then follow a singular perturbation approach to prove also Algorithm \ref{alg:MRARL} solves the robustly stable on-policy data-driven LQR problem.
		
	\subsection{Stability Result for the Reduced-Order System}
		Consider the Model Reference Adaptive Reinforcement Learning with ARE implementation of $\hP$ as in \eqref{eq:ARE_static}.
		Following the notation of Section \ref{sec:problem_setup}, the controller obtained by combining the value function identifier \eqref{eq:swapping}, \eqref{eq:identifier}, \eqref{eq:ARE_static}, reference model \eqref{eq:model_reference}, and the adaptive stabilizer \eqref{eq:adaptive_dynamics}, \eqref{eq:system_input} is in the form \eqref{eq:controller_generic}, with state
		\begin{equation}\label{eq:z_s}
			z_{\textup{s}} \coloneqq (\xi, \zeta, \hat{A}, x_{\text{m}}, \hKe) \in \mathcal{Z}_{\textup{s}} \coloneqq \mathbb{R}^{2\dimx}\times\Theta\times\mathbb{R}^{\dimx}\times\mathbb{R}^{\dimu\times\dimx},
		\end{equation}
		and output policy
		\begin{equation}
			\pi(x, z_{\textup{s}}, d) \coloneqq (\hKe - BR^{-1} B^{\top} \mathcal{P}(\hA)) x +d,
		\end{equation}
		from which it follows that the learning set $\mathcal{L}$ in \eqref{eq:learning_set} is non-empty and given by
		\begin{equation}\label{eq:L_s}
			\mathcal{L}_{\textup{s}} \coloneqq \{z_{\textup{s}} \in \mathcal{Z}_{\textup{s}}: \hKe - BR^{-1} B^{\top} \mathcal{P}(\hA) = - BR^{-1} B^{\top} P^\star\}.
		\end{equation}
		In particular, we recall that our algorithm aims at reaching the learning set by ensuring $\hA \to A$ (hence $\hP \equiv \mathcal{P}(\hA) \to P^\star$ by continuity of the map $\mathcal{P}(\hA)$) and $\hKe \to 0$.
		The following result shows that, with $\gamma > 0$ sufficiently small, the Model Reference Adaptive Reinforcement Learning with ARE implementation of $\hP$ solves the robustly stable on-policy data-driven LQR problem.
		\begin{theorem}\label{theorem_ARE}
			Consider the closed-loop system given by the interconnection of plant \eqref{eq:plant_dynamics} and the controller of Algorithm~\ref{alg:MRARL}, with $\hP(t)=\mathcal{P}(\hA(t))$ for all $t$ and $\Pta$ satisfying \eqref{eq:ARE_static}.
			Let the stationary dither $d$ be generated by an exosystem of the form \eqref{eq:exosystem} and let its entries be sufficiently rich of order $\dimx + 1$ and uncorrelated.
			Then, there exists $\gamma^\star > 0$ such that, for all $\gamma \in (0, \gamma^\star]$, there exists a compact set $\mathcal{A}_s$ satisfying
			\begin{equation}\label{eq:omega_set_of_reduced}
				\begin{split}
					\mathcal{A}_{\textup{s}} \subset \{(w, x, z_{\textup{s}})\in &\;\W\times\mathbb{R}^{\dimx}\times\mathcal{L}_{\textup{s}}: \hA = A, \hKe = 0, x = x_{\textup{m}}, \epsilon=0\}
				\end{split}
			\end{equation}
			that is uniformly globally asymptotically stable.
		\end{theorem}
	
	\subsection{Stability Result for MR-ARL}
		Consider the Model Reference Adaptive Reinforcement Learning (MR-ARL) algorithm with DRE implementation of $\hP$ as in \eqref{eq:DRE_implementation} (Algorithm \ref{alg:MRARL}).
		Following the notation of Section \ref{sec:problem_setup}, the controller obtained by combining the value function identifier \eqref{eq:swapping}, \eqref{eq:identifier}, \eqref{eq:DRE_implementation}, reference model \eqref{eq:model_reference}, and the adaptive stabilizer \eqref{eq:adaptive_dynamics}, \eqref{eq:system_input} is in the form \eqref{eq:controller_generic}, with state
		\begin{equation}
			z \coloneqq (z_{\textup{s}}, \hP) \in \mathcal{Z} \coloneqq \mathcal{Z}_{\textup{s}}\times \Psdef,
		\end{equation}
		where $z_{\textup{s}}$ and $\mathcal{Z}_{\textup{s}}$ are given in \eqref{eq:z_s}.
		The output policy then becomes
		\begin{equation}
			\pi(x, z, d) \coloneqq (\hKe - BR^{-1} B^{\top} \hP) x +d,
		\end{equation}
		and the learning set $\mathcal{L}$ is given by
		\begin{equation}\label{eq:learning_set_mrarl}
			\Lz\coloneqq \{z \in \mathcal{Z}: \hKe - BR^{-1} B^{\top} \hP = - BR^{-1} B^{\top} P^\star\}.
		\end{equation}
		In this case, the learning set is reached via $\hA \to A$ (hence $\hP \to P^\star$ by asymptotic stability of the DRE \eqref{eq:DRE}) and $\hKe \to 0$.
		The next result, which is the main result of this work, shows that with $\gamma > 0$ sufficiently small and $g > 0$ sufficiently large, the Model Reference Adaptive Reinforcement Learning in Algorithm \ref{alg:MRARL} solves the robustly stable on-policy data-driven LQR problem.
		
		\begin{theorem}\label{theorem_DRE}
			Consider the closed-loop system given by the interconnection of plant \eqref{eq:plant_dynamics} and the controller of Algorithm \ref{alg:MRARL}.
			Pick $d$ and $\gamma$ as in Theorem \ref{theorem_ARE}. 
			Then, the compact set 
			\begin{equation}
				\begin{split}
					\mathcal{A}\!\coloneqq\! \mathcal{A}_{\textup{s}}\times P^\star\! &\subset  \{(w, x, z)\in \W\times\mathbb{R}^{\dimx}\times\mathcal{L}:\hA = A, \hKe = 0, x = x_{\textup{m}}, \epsilon=0, P = P^\star\}
				\end{split}
			\end{equation}
			is semiglobally uniformly asymptotically stable in the tuning parameter $g>0$, where $\mathcal{A}_{\textup{s}}$ is given in \eqref{eq:omega_set_of_reduced} and $g$ is the one in \eqref{eq:DRE_implementation}.
			Namely, for any compact set $\mathcal{K} \subset \W\times\mathbb{R}^{\dimx}\times\mathcal{Z}$ of initial conditions for the closed-loop system, there exists $g > 0$ such that $\mathcal{A}$ is uniformly asymptotically stable with domain of attraction containing $\mathcal{K}$.
		\end{theorem}

\section{Algorithm Analysis}\label{sec:analysis}

	In the following, we will only study the properties of the reduced-order version of the algorithm, i.e., with $\hP(t)=\mathcal{P}(A(t))$ for all $t$ and $\Pta$ satisfying \eqref{eq:ARE_static}.
	The second result, i.e., the stability of Algorithm~\ref{alg:MRARL} (implementing the DRE), is obtained by invoking singular perturbations techniques. 
	
	\subsection{Error Dynamics}\label{sec:erros_system}
	
		We begin the analysis by presenting the closed-loop dynamics in error coordinates, which is used to provide the technical results of the following subsections.
		
		\subsubsection{Identifier Dynamics}
		
			Consider the error coordinate $\tilde{\epsilon}$ in \eqref{eq:tilde_eps}, which can be written as
			\begin{equation}
				\begin{split}
					\tilde{\epsilon} & \coloneqq A\xi - (x + \zeta).
				\end{split}
			\end{equation}
			Then, from \eqref{eq:plant_dynamics}, \eqref{eq:swapping}, it holds that
			\begin{equation}\label{eq:tilde_eps_dot}
				\begin{split}
					\dot{\tilde{\epsilon}} &= A(-\lambda \xi + x) - (Ax + Bu - \lambda(x + \zeta) - Bu)\\
					&= -\lambda(A\xi - (x + \zeta)) = -\lambda \tilde{\epsilon},
				\end{split}
			\end{equation}
			which ensures that the prediction error 
			$\epsilon \coloneqq \hat{A}\xi - (x + \zeta) =(\hat{A} - A)\xi + \tilde{\epsilon}$ converges to $(\hA - A)\xi$ exponentially.
			\par Define $\tA \coloneqq \hA - A$. Then, from \eqref{eq:tilde_eps}, \eqref{eq:tilde_eps_dot}, we can rewrite the identifier dynamics \eqref{eq:swapping}, \eqref{eq:identifier}, \eqref{eq:prediction_error} in error coordinates as the following cascaded system
			\begin{equation}\label{eq:id_error_dyn}
				\begin{split}
					\dot{\tilde{\epsilon}} &= -\lambda \tilde{\epsilon}\\
					\dtA &=\! \proj_{\hA \in \mathcal{C}} \!
					\left\{- \gamma BB^\dagger \frac{\tA \xi\xi^\top + \tilde{\epsilon}\xi^\top}{1 + \nu |\xi| |\epsilon|}\right\},
				\end{split}
			\end{equation}
			driven by $\xi(t)$, solution of the filter
			\begin{equation}\label{eq:xi_dyn}
				\dot{\xi} = -\lambda \xi + x.
			\end{equation}
			
			\begin{remark}\label{remark:PE_necessity}
				To ensure $\hA(t)\to A$, it is known from the adaptive control literature that vector $\xi(t)$ must be a persistently exciting (PE) signal \cite{ioannou1996robust}.
				However, notice that $\xi(t)$ is a filtered version of $x(t)$, which is generated in closed-loop by interconnecting the plant and the controller. For this reason, special care will be dedicated to its analysis.
			\end{remark}
			
		\subsubsection{Reference Model Dynamics}
		
			From \eqref{eq:ARE_static}, when $\hP = \Pta$, system \eqref{eq:model_reference} can be written highlighting the dependence on the estimate $\hA$ of the identifier:
			\begin{equation}\label{eq:model_ref_detailed}
				\dot{x}_{\text{m}} = (\hA- B R^{-1} B^\top \Pta)\xm + Bd,
			\end{equation}
			where from \eqref{eq:identifier}, \eqref{eq:ARE_static}, the pointwise-in-time value of $\Pta$ is provided implicitly as the solution of a parameter-varying ARE.
			By \cite[Thm. 4.1]{ran1988parameter}, $\Pta$ is an analytic function of $\hA$, being all matrices of ARE $\mathcal{R}(P, \hA)=0$ in \eqref{eq:ARE_implementation} analytic functions of $\hA\in \KAL$.
			From this fact, matrix $\hA- B R^{-1} B^\top \Pta$ is Hurwitz and an analytic function of $\hA$.
			
		\subsubsection{Adaptive Tracking Dynamics} 
		
			We conclude this overview by studying the interconnection of the error dynamics \eqref{eq:error_dynamics_2} and the adaptive controller \eqref{eq:adaptive_dynamics}, \eqref{eq:system_input}.
			We define $\tKe\coloneqq \hKe - \Ke$.
			By choosing~\eqref{eq:system_input} as input for~\eqref{eq:error_dynamics}, we obtain:
			\begin{equation}\label{eq:error_dyn}
				\begin{split}
					\de = (\hA-BR^{-1}B^\top \Pta) e + B(\hKe x - \Kta x)
				\end{split}
			\end{equation}
			By choosing expression \eqref{eq:explicit_Ke} for $\Ke$, we can explicitly calculate the variation in time of $\Ke$ due to the movement of $\hA$.
			This is out of the standard framework of model reference adaptive control, and thus particular attention is required.
			We can calculate the time derivative of $\Ke$ by deriving \eqref{eq:explicit_Ke}:
			\begin{equation}\label{eq:mrac_drift}
				\begin{split}
					\dKe &= B^\dagger \dhA.
				\end{split}
			\end{equation}
			Since both $B$ and $\dhA$ are known, we can use their knowledge to implement adaptive law \eqref{eq:adaptive_dynamics}, which takes into account this drift.
			Given equations \eqref{eq:adaptive_dynamics} and \eqref{eq:mrac_drift}, the induced dynamics for $\tKe$ is:
			\begin{equation}\label{eq:adaptive_error_system}
				\begin{split}
					\dtKe &= \dhKe - \dKe \\
					&= -\mu B^\top \Pta (x-\xm) x^\top + B^\dagger \dhA - B^\dagger \dhA\\
					&= -\mu B^\top \Pta e x^\top.
				\end{split}
			\end{equation}
		
	\subsection{Global Boundedness of Solutions}\label{sec:boundedness}
	
		We now show boundedness and forward completeness of the solutions of the closed-loop system obtained from the interconnection of the identifier dynamics \eqref{eq:id_error_dyn}, \eqref{eq:xi_dyn}, the reference model \eqref{eq:model_ref_detailed}, and the adaptive error system \eqref{eq:error_dyn}, \eqref{eq:adaptive_error_system}.
		The overall analysis entails proving uniform bounds on the solutions of the main involved subsystems, then combining the results using arguments similar to \cite[Thm. 6.3]{krstic1995nonlinear} (see the proof of Proposition \ref{lemma:interconnection_boundedness}).
		To increase readability, we leave the proofs of the technical lemmas in the Appendix.
		\par We begin by showing uniform boundedness of $\hA$ and $\dhA$.
		\begin{lemma}\label{lemma:bounded_id}
			Let the maximal interval of solutions of \eqref{eq:id_error_dyn}, \eqref{eq:xi_dyn}, \eqref{eq:model_ref_detailed}, \eqref{eq:error_dyn}, \eqref{eq:adaptive_error_system} be $[0, t_f)$.
			Then, it holds that 
			\begin{enumerate}[I]
				\item) $\tilde{\epsilon}(\cdot), \tA(\cdot)$ are uniformly bounded in the interval $[0, t_f)$
				\item) $\hA(t) \in \KAL$ for all $t \in [0, t_f)$
			\end{enumerate}
			Furthermore, if $t_f = \infty$, the origin $(\tilde{\epsilon},\tA) = 0$ of system \eqref{eq:id_error_dyn}, driven by input $\xi(t)$, is uniformly globally stable (UGS).
		\end{lemma}
		\begin{lemma}\label{lemma:bounded_A_dot}
			Let the maximal interval of solutions of \eqref{eq:id_error_dyn}, \eqref{eq:xi_dyn}, \eqref{eq:model_ref_detailed}, \eqref{eq:error_dyn}, \eqref{eq:adaptive_error_system} be $[0, t_f)$.
			Then, it holds that
			\begin{equation}
				|\dhA(t)| \leq \gamma, \qquad \forall t \in [0, t_f).
			\end{equation}
		\end{lemma}
		\begin{remark}
			The above results hold even if the input $\xi(t)$ of the identifier escapes to infinity as $t \to t_f$.
		\end{remark}
		Although the overall boundedness analysis entails also the study of $\xi(t)$, system \eqref{eq:xi_dyn} ISS with respect to input $x(t)$, thus its behavior will be analyzed directly in Proposition \ref{lemma:interconnection_boundedness}.
		
		Then, we show that the reference model \eqref{eq:model_ref_detailed} is bounded as long as $|\dhA(t)|$ is sufficiently small.
		\begin{lemma}\label{lemma:MR_ISS}
			Let the maximal interval of solutions of \eqref{eq:id_error_dyn}, \eqref{eq:xi_dyn}, \eqref{eq:model_ref_detailed}, \eqref{eq:error_dyn}, \eqref{eq:adaptive_error_system} be $[0, t_f)$.
			There exists $\gamma^\star_b > 0$ such that, if $|\dhA(t)| \leq \gamma^\star_b$ for all $t \in [0, t_f)$, then $\xm(\cdot)$ is uniformly bounded over the interval $[0, t_f)$.
			Furthermore, if $t_f = \infty$, then the reference model \eqref{eq:model_ref_detailed} with input $d(t)$ is input-to-state stable.
		\end{lemma}
		Next, we provide a statement for system \eqref{eq:error_dyn}, \eqref{eq:adaptive_error_system}.
		\begin{lemma}\label{lemma:AC_ugs}
			Let the maximal interval of solutions of \eqref{eq:id_error_dyn}, \eqref{eq:xi_dyn}, \eqref{eq:model_ref_detailed}, \eqref{eq:error_dyn}, \eqref{eq:adaptive_error_system} be $[0, t_f)$.
			Pick $\gamma^\star_b>0$ from Lemma~\ref{lemma:MR_ISS} and let $|\dhA(t)| \leq \gamma^\star_b$ for all $t \in [0, t_f)$.
			Then, signals $e(\cdot), \tKe(\cdot)$ are uniformly bounded in the interval $[0, t_f)$.
			Furthermore, if $t_f = \infty$, the origin $(e, \tKe) = 0$ of system \eqref{eq:error_dyn}, \eqref{eq:adaptive_error_system}, with input $\hA(t)$, is UGS. 
		\end{lemma}
		Finally, we combine the previous results to obtain that solutions are globally bounded and forward complete.
		\begin{proposition}\label{lemma:interconnection_boundedness}
			Consider the closed-loop system obtained from the interconnection of the identifier dynamics \eqref{eq:id_error_dyn}, \eqref{eq:xi_dyn}, the reference model \eqref{eq:model_ref_detailed}, and the adaptive error system \eqref{eq:error_dyn}, \eqref{eq:adaptive_error_system}. Pick $\gamma^\star_b$ from Lemma \ref{lemma:bounded_A_dot}.
			If $\gamma \in (0, \gamma^\star_b]$, then the closed-loop solutions are bounded and forward complete.
		\end{proposition}
		\begin{proof}
			Suppose that the maximal interval of existence of the solution of \eqref{eq:id_error_dyn}, \eqref{eq:xi_dyn}, \eqref{eq:model_ref_detailed}, \eqref{eq:error_dyn}, and \eqref{eq:adaptive_error_system} is $[0, t_f)$.
			Then, from Lemma \ref{lemma:bounded_id}, $\tA(\cdot)$ and $\tilde{\epsilon}(\cdot)$ are uniformly bounded. 
			From Lemma \ref{lemma:bounded_A_dot}, $|\dhA(\cdot)|$ is uniformly bounded by $\gamma$.
			Consider any $\gamma \in (0, \gamma^\star_b]$, then Lemmas \ref{lemma:MR_ISS} and \ref{lemma:AC_ugs} ensure that $\xm(\cdot)$, $e(\cdot)$, and $\tKe(\cdot)$ are uniformly bounded, thus also $\xi(\cdot)$ is uniformly bounded from \eqref{eq:xi_dyn} and standard ISS results.
			
			We have thus shown that all signals of the closed-loop system are bounded, with bounds that do not depend on $t_f$.
			By contradiction, we conclude that $t_f = \infty$, thus the solutions are forward complete.
			Namely, if $t_f$ were finite, the solutions would leave any compact set as $t \to t_f$, contradicting the independence of the bounds on $t_f$ \cite[Thm. 6.3]{krstic1995nonlinear}.
		\end{proof}
	
	\subsection{Exponential Convergence to the Optimal Policy}\label{sec:convergence}
	
		We now focus on the uniform asymptotic stability properties of the closed-loop system \eqref{eq:id_error_dyn}, \eqref{eq:xi_dyn}, \eqref{eq:model_ref_detailed}, \eqref{eq:error_dyn}, \eqref{eq:adaptive_error_system}.
		First, we show that $\xm(t)$ is persistently exciting as long as $|\dhA|$ is sufficiently small.
		\begin{lemma}\label{lemma:x_m_uPE}
			Let the entries of stationary input $d$ be sufficiently rich of order $n + 1$ and uncorrelated.
			There exists $\gamma^\star_{PE} \in (0, \gamma^\star_b]$, with $\gamma_b^\star$ from Proposition \ref{lemma:interconnection_boundedness}, such that, for all $\gamma \in (0, \gamma^\star_{PE}]$, the solutions $x_{\textup{m}}(t)$ of the reference model \eqref{eq:model_ref_detailed} are persistently exciting (PE).
		\end{lemma}
		Next, we provide a direct consequence of Lemma \ref{lemma:x_m_uPE} for the adaptive error dynamics \eqref{eq:error_dyn}, \eqref{eq:adaptive_error_system}.
		\begin{lemma}\label{lemma:AC_UGAS}
			Let the hypotheses of Lemma \ref{lemma:x_m_uPE} hold and let $\gamma \in (0, \gamma^\star_{PE}]$, where $\gamma^\star_{PE}$ is given in Lemma \ref{lemma:x_m_uPE}.
			Then, the origin $(e, \tKe) = 0$ of system \eqref{eq:error_dyn}, \eqref{eq:adaptive_error_system} is uniformly globally asymptotically stable (UGAS) and uniformly locally exponentially stable (ULES).
		\end{lemma}
		Now that we have established that every solution $e(t)$ converges exponentially to zero, uniformly from compact sets of initial conditions, we can conclude the convergence analysis by studying the identifier dynamics \eqref{eq:id_error_dyn}.
		\begin{lemma}\label{lemma:id_UGES}
			Let the hypotheses of Lemma \ref{lemma:x_m_uPE} hold and let
			$\gamma \in (0, \gamma^\star_{PE}]$, where $\gamma^\star_{PE}$ is given in Lemma \ref{lemma:x_m_uPE}.
			Then, the origin $(\tilde{\epsilon}, \tA) = 0$ of system \eqref{eq:id_error_dyn}, with input $\xi(t)$, is uniformly globally exponentially stable (UGES).
		\end{lemma}
	
	\subsection{Proof of the Main Results}
	
		\subsubsection{Proof of Theorem \ref{theorem_ARE}}
		
			Pick $\gamma^\star\coloneqq \min\{\gamma_b^\star, \gamma_{PE}^\star\}=\gamma_{PE}^\star$, where $\gamma_b^\star$ is from Proposition \ref{lemma:interconnection_boundedness} and $\gamma_{PE}^\star$ is the one of Lemma \ref{lemma:x_m_uPE}.
			Then, if $\gamma\leq \gamma^\star$, the closed-loop solutions are bounded and forward complete.
			Moreover, $\xm$ is PE.
			The remainder of the proof involves showing the existence of a UGAS attractor using the concept of $\omega$-limit set of a set, see \cite[Def. 6.23]{goebel2012hybrid}. 
			
			By Lemmas \ref{lemma:AC_UGAS} and \ref{lemma:id_UGES}, from any compact set of initial conditions, it holds that $\hA \to A, \tilde{\epsilon}\to 0, e \to 0, \tKe \to 0$ exponentially. Moreover, by Lemma \ref{lemma:MR_ISS}, the model reference subsystem \eqref{eq:model_ref_detailed} is ISS with uniformly bounded input $d(t)$, in particular we have that:
			\begin{equation}
				|\xm| \geq X_{\text{m}} \implies \dot{V}_{\text{m}} \leq -\frac{q}{4}|\xm|^2,
			\end{equation}
			where $X_{\text{m}}$ can be found in \eqref{eq:ultimate_bound_ref_model} and depends on $\|d(\cdot)\|_\infty$, and $V_{\text{m}}$ in \eqref{eq:Lyapunov_for_ref_model} is an ISS Lyapunov function for the reference model. 
			Consider the $\xi$ subsystem in \eqref{eq:swapping}.
			It holds that
			\begin{equation}\label{eq:Lyapunov_for_xi}
				|\xi| \geq \frac{2|x|}{\lambda}\implies \frac{\text{d}}{\text{d}t}\left(\frac{1}{2}|\xi|^2\right)\leq -\frac{\lambda}{2}|\xi|^2.
			\end{equation}
			Denote $\Xi \coloneqq \frac{2}{\lambda}X_\text{m}$ and define the compact set 
			\begin{equation}\label{eq:K_s_star}
				\begin{split}
					\mathcal{K}_{\textup{s}}^\star\coloneqq \{&(w, x, z_{\textup{s}}) \in \mathcal{W}\times\mathbb{R}^{\dimx}\times\mathcal{Z}_{\textup{s}}:\hA = A, \tilde{\epsilon}=0, e=0, |\xm|\leq X_\text{m}, |\xi|\leq \Xi\}\subset \W\times\mathbb{R}^n\times\mathcal{L}_{\textup{s}},
				\end{split}
			\end{equation}
			where $\mathcal{L}_{\textup{s}}$ is the learning set given in \eqref{eq:L_s}. 
			Consider a set of initial conditions $\mathcal{K}_{s} \coloneqq \mathcal{K}_s^\star + c\mathbb{B}$, with $c > 0$ arbitrary, and note that the solutions are empty if they start outside $\W\times \R^n\times\Z_{\textup{s}}$.
			We now prove that $\mathcal{K}_s^\star$ is uniformly attractive from $\mathcal{K}_s$. 
			By the above-mentioned properties for the subsystems $\tA, \tilde{\epsilon}, e, \tKe, \xm$, there exists $T^\prime>0$ such that, for any $\varepsilon>0$, it holds that
			\begin{equation}
					|\tA(t)|\leq \varepsilon, \;\;\; |\tilde{\epsilon}(t)|\leq \varepsilon, \;\;\;|e(t)|\leq \min\left(\varepsilon, \frac{\lambda}{2}\frac{\varepsilon}{3}\right),\;\;\;|\xm(t)| \leq X_\text{m} + \min\left(\varepsilon, \frac{\lambda}{2}\frac{\varepsilon}{3}\right).
			\end{equation}
			for all $t\geq T^\prime$, from which it holds also that
			\begin{equation}
				\begin{split}
					\frac{2|x(t)|}{\lambda} &\leq \frac{2}{\lambda}(|\xm(t)|+|e(t)|)\\
					&\leq \Xi + \frac{\varepsilon}{3} + \frac{\varepsilon}{3} \leq \Xi + \frac{2}{3}\varepsilon.
				\end{split}
			\end{equation}
			Thus, from \eqref{eq:Lyapunov_for_xi}, there exists $T \geq T^\prime$ such that $|\xi(t)|\leq \Xi + \varepsilon$ for all $t\geq T$.
			
			For compactness of notation, denote $\boldsymbol{x}_{\textup{s}} \coloneqq (w, x, z_{\text{s}})$.
			The arguments above have proved that $\mathcal{K}_s^\star$ is uniformly attractive from $\mathcal{K}_s$.
			Namely, for any $\rho > 0$, there exists $T_\rho \geq 0$ such that $|\phi(t, \boldsymbol{x}_{\textup{s}})|_{\mathcal{K}_{\textup{s}}^\star} \leq \rho$, for all $t \geq T_{\rho}$ and $\boldsymbol{x}_{\textup{s}} \in \mathcal{K}_{\textup{s}}$, where $\phi(t, \boldsymbol{x}_{\textup{s}})$ is the solution at time $t$ of the closed-loop system having initial condition $\boldsymbol{x}_{\textup{s}}$.
			
			Denote with $\As\coloneqq \Omega(\mathcal{K}_\textup{s})$ the $\omega$-limit set of $\mathcal{K}_\textup{s}$.
			We want to prove that $\As \subset \mathcal{K}_{\textup{s}}^\star$.
			We do it by contradiction, i.e., we suppose that $\As \subset \mathcal{K}_{\textup{s}}^\star$ is false.
			Under this hypothesis, there exists $\bar{\boldsymbol{x}}_{\textup{s}} \in \As$ and $\rho > 0$ such that $|\bar{\boldsymbol{x}}_{\textup{s}}|_{\mathcal{K}_{\textup{s}}^\star} \geq 3\rho$.
			By definition \cite[Def. 6.23]{goebel2012hybrid}, the $\omega$-limit set of $\mathcal{K}_\textup{s}$ is the set of all points $\boldsymbol{x}_{\textup{s}}$ such that there exist sequences $\boldsymbol{x}_{\textup{s}, n} \in \mathcal{K}_\textup{s}$, $t_n \geq 0$ such that $\lim_{n \to \infty}t_n = \infty$ and $\lim_{n \to \infty}\phi(t_n, \boldsymbol{x}_{\textup{s}, n}) = \boldsymbol{x}_{\textup{s}}$.
			Therefore, by definition of limit, there exists $\bar{n} \in \mathbb{N}$ such that
			\begin{equation}
				|\phi(t_n, \boldsymbol{x}_{\textup{s}, n}) - \bar{\boldsymbol{x}}_{\textup{s}}|\leq \rho, \quad \forall n \geq \bar{n}.
			\end{equation}
			Pick any subsequence $\boldsymbol{x}_{\textup{s}, {n_i}}$, $t_{n_i}$ such that, for $n_i \geq \bar{n}$, then $t_{n_i} \geq T_{\rho}$, where $T_{\rho}$ derives from the uniform attractivity of $\mathcal{K}_{\textup{s}}^\star$ (see above).
			We have thus proved that, for $n_i \geq \bar{n}$, $|\phi(t_{n_i}, \boldsymbol{x}_{\textup{s}, n_i}) - \bar{\boldsymbol{x}}_{\textup{s}}|\leq \rho$, thus $|\phi(t_{n_i}, \boldsymbol{x}_{\textup{s}, n_i})|_{\mathcal{K}_{\textup{s}}^\star} \geq 2\rho$, and at the same time $|\phi(t_{n_i}, \boldsymbol{x}_{\textup{s}, n_i})|_{\mathcal{K}_{\textup{s}}^\star} \leq \rho$ by uniform attractivity of $\mathcal{K}_{\textup{s}}^\star$.
			This is a contradiction, hence necessarily $\As \subset \mathcal{K}_{\textup{s}}^\star$.
			
			To summarize the previous results, we have thus proved that the solutions are globally bounded and forward complete and
			\begin{equation}
				\As\coloneqq \Omega(\mathcal{K}_{\textup{s}})\subset \mathcal{K}_{\textup{s}}^\star\subset \Int(\mathcal{K}_{\textup{s}})\subset\mathcal{K}_{\textup{s}}.
			\end{equation}
			By \cite[Corollary $7.7$]{goebel2012hybrid}, $\As=\Omega(\mathcal{K}_\textup{s})$ is asymptotically stable, with domain of attraction containing $\mathcal{K}_{\textup{s}}$.
			Since $\mathcal{K}_{\textup{s}}$ can be chosen arbitrarily large due to Proposition \ref{lemma:interconnection_boundedness}, we conclude UGAS of $\As$.
	
		\subsubsection{Proof of Theorem \ref{theorem_DRE}}
		
			Consider the reduced-order system with state $\bm{x}_{\text{s}}\coloneqq (w, x, z_{\text{s}})$ and the boundary layer system with state $\hP$.
			Define the indicator functions
			\begin{equation}
				\begin{split}
					\omega_s (\bm{x}_{\text{s}}) &\coloneqq 
					\begin{cases}
						|\bm{x}_{\text{s}}|_{\As} & \qquad\;\bm{x}_{\text{s}} \in \W \times \R^\dimx \times \mathcal{Z}_{\textup{s}}\\
						\infty & \qquad\; \text{elsewhere}
					\end{cases}\\
					\omega_f (\bm{x}_{\text{s}}, \hP) &\coloneqq 
					\begin{cases}
						|\hP-\mathcal{P}(\hA)| & (\bm{x}_{\text{s}}, \hP) \in \W \times \R^\dimx \times \mathcal{Z}\\
						\infty & \text{elsewhere}.
					\end{cases}\\
				\end{split}
			\end{equation}
			By Theorem \ref{theorem_ARE}, the reduced-order system satisfies
			\begin{equation}
				\omega_s (\bm{x}_{\text{s}}(t)) \leq \beta_s (\omega_s (\bm{x}_{\text{s}}(0)), t),
			\end{equation}
			where $\beta_s$ is a class $\mathcal{KL}$ function.
			Moreover, by the DRE properties \cite[Thm. 4]{bucy1967global}, the boundary-layer system $d\hat{P}/d\tau(\tau) = \mathcal{R}(\hP(\tau), \hA)$, with $\hA \in \KAL$ constant and $\tau \coloneqq gt$, satisfies
			\begin{equation}
				\omega_f (\bm{x}_{\text{s}}, \hP(\tau))\leq \beta_f(\omega_f (\bm{x}_{\text{s}}, \hP(0)), \tau),
			\end{equation}
			where $\beta_f$ is a class $\mathcal{KL}$ function.
			From \cite[Thm. 1]{teel2003unified} (Assumptions $1, 3, 4, 7, 8$ can be verified), from any compact set of initial conditions $\mathcal{K} \subset \W \times \R^\dimx\times \Z$ and for any $\delta > 0$, there exists $g^\star>0$ such that, for all $g \geq g^\star$, the solutions are forward complete and satisfy:
			\begin{equation}\label{eq:singular_perturbation}
				\begin{split}
					\omega_s (\bm{x}_{\text{s}}(t)) &\leq \beta_s (\omega_s (\bm{x}_{\text{s}}(0)), t) + \delta\\
					\omega_f (\bm{x}_{\text{s}}(t), \hP(t))&\leq \beta_f(\omega_f (\bm{x}_{\text{s}}(0), \hP(0)), gt) + \delta.
				\end{split}
			\end{equation}
			In particular, choose $\mathcal{K}\coloneqq \As\times P^\star + c\mathbb{B}$, with $c > 0$ arbitrary.
			Reference model dynamics \eqref{eq:model_reference} can be rewritten as:
			\begin{equation}\label{eq:perturbed_model_reference}
				\begin{split}
					\dxm =& (\hA - BR^{-1}B^\top \hP)\xm + Bd\\
					=& (A - BR^{-1}B^\top P^\star)\xm + Bd + (\tA - BR^{-1}B^\top (\hP-P^\star))\xm\\
					=& (A - BR^{-1}B^\top P^\star)\xm + Bd + (\tA - BR^{-1}B^\top (\hP-\Pta + \Pta - \mathcal{P}(A)))\xm .
				\end{split}
			\end{equation}
			Notice that $\xm$ is PE if $\hP = \Pta$ and $\tA = 0$.
			Furthermore, since $\mathcal{P}(\hA)$ is an analytic function of $\hA$, $\xm$ is PE by \cite[Lemma 6.1.2]{sastry1990adaptive} if $|\hP - \Pta|$ and $|\tA|$ are sufficiently small, because the solutions of  \eqref{eq:perturbed_model_reference} are sufficiently close to those with $\hP = \Pta$ and $\tA = 0$.
			Moreover, also $x$ and $\xi$ are PE if $\xm$ is PE and $|e|$ is sufficiently small.
			Choose $\delta > 0$ such that the conditions $\omega_s (\bm{x}_{\text{s}}) \leq 2\delta$ and $\omega_f(\bm{x}_\textup{s}, \hP)\leq 2\delta$ imply that $\xm$, $x$, and $\xi$ are PE.
			Then, pick $g \geq g^\star$, where $g^\star$ is obtained from the considered $\mathcal{K}$ and $\delta$.
			From \eqref{eq:singular_perturbation}, the closed-loop solutions converge in finite time $T$ to a compact set satisfying $\omega_s (\bm{x}_{\text{s}}) \leq 2\delta$ and $\omega_f(\bm{x}_\textup{s}, \hP)\leq 2\delta$.
			Then, for $t \geq T$, $\tilde{\epsilon}\to 0, \hA \to A$ exponentially from Lemma \ref{lemma:id_UGES} since $\xi$ is PE.
			From the local exponential stability of the DRE \cite[Thm. 4]{bucy1967global}, it follows that $\hP \to P^\star$ exponentially.
			By Lemma \ref{lemma:AC_UGAS} and $\hP \to \Pta$, we conclude that $e \to 0, \hKe \to 0$ exponentially.
			As a consequence, the same arguments of Theorem \ref{theorem_ARE} (omitted here to avoid repetition) can be used to show that the compact set
			\begin{equation}\label{eq:mathK_star}
				\begin{split}
					\mathcal{K}^\star&\coloneqq \{(\bm{x}_{\text{s}}, \hP) \in \W \times \R^\dimx\times \Z : \hA = A, 
					\tilde{\epsilon}=0, e=0, |\xm|\leq X_\text{m}, |\xi|\leq \Xi, \hP=P^\star \} \\
					&= \mathcal{K}^\star_s\times P^\star,		
				\end{split}
			\end{equation}
			is uniformly attractive from $\mathcal{K}$, with $\mathcal{K}^\star_s$ given in \eqref{eq:K_s_star} . The same steps as in Theorem \ref{theorem_ARE} allow to prove that $\A\coloneqq\Omega(\mathcal{K})\subset \mathcal{K}^\star\subset \Int(\mathcal{K})\subset \mathcal{K}$, thus $\A$ is uniformly asymptotically stable with domain of attraction containing $\mathcal{K}$, and since $\mathcal{K}$ can be chosen arbitrarily large we can conclude semiglobal uniform asymptotic stability of $\A$.
			
			Finally, we want to prove that $\A = \As \times P^\star$, where $\As=\Omega(\mathcal{K}^\star_s)$.
			In $\A$, it holds that $\hP = P^\star$, $\hA = A$ and $\tilde{\epsilon}=0$, from which it holds that $\hP = \mathcal{P}(\hA)=\mathcal{P}(A)=P^\star$ for all points in this set. 
			For this reason, in $\A$, the vector field of Algorithm~\ref{alg:MRARL} coincides with the vector field of the reduced-order system with $\hA = A$ and $\tilde{\epsilon}=0$. 
			Since the vector fields coincide, we have that in this set solutions $\bm{x}(t)$ of Algorithm \ref{alg:MRARL} can be written as $\bm{x}(t)=\bm{x}_{\text{s}}(t)\times P^\star$, where $\bm{x}_s(t)$ is the solution of the reduced-order system having the same initial conditions.
			
			From \eqref{eq:mathK_star} and $\A \subset \mathcal{K}^\star$, it follows that $\A = \Omega(\mathcal{K}) = \Omega (\mathcal{K}^\star) = \Omega (\mathcal{K}_s^\star \times P^\star) = \mathcal{A}^\prime \times P^\star$. Since for the slow states of Algorithm 1 the solutions coincide with those of the reduced-order system, it follows that $\mathcal{A}^\prime = \As$. As a consequence, $\A = \As \times P^\star$.

\section{Numerical Analysis: Control of a Doubly
	Fed Induction Motor}\label{sec:example}

	In this section, we propose two numerical examples to show the effectiveness of Model Reference Adaptive Reinforcement Learning.
	In the first example, we consider the model of a doubly fed induction motor (DFIM) at constant speed with unknown rotor and stator resistances.
	In the second example, we test the robustness of the proposed algorithm by considering a DFIM with slowly time-varying unknown resistances, due to the motor heating up, and rotor acceleration.
	
	\subsection{Example 1: Constant Parameters}
	
		A DFIM at constant speed can be modeled \cite{leonhard2001control} with a linear system in the form of \eqref{eq:plant_dynamics} with state
		\begin{equation}
			x = (i_{1u}, i_{1v}, i_{2u},i_{2v}) \in \R^4,
		\end{equation}
		where $i_{1u}, i_{1v}$ are the stator currents and $i_{2u}, i_{2v}$ are the rotor currents. The input is
		\begin{equation}
			u = (u_{1u}, u_{1v}, u_{2u}, u_{2v}) \in \R^4,
		\end{equation}
		where $u_{1u}, u_{1v}$ are the stator voltages and $u_{2u}, u_{2v},$ the rotor voltages. System matrices are defined as
		\begin{equation}\label{eq:motor_matrices}
				A = \frac{1}{\bar{L}}
				\begin{bmatrix}
					-L_2R_1 &  -\alpha + \beta & L_mR_2 & \beta_2\\
					\alpha - \beta & -L_2R_1 & -\beta_2 & -L_mR_2 \\
					L_mR_1 & -\beta_1 & -L_1R_2 & -\alpha-\beta_{12}\\
					\beta_1 & L_mR_1 &  \alpha+\beta_{12}& -L_1R_2
				\end{bmatrix}, \qquad B = \frac{1}{\bar{L}}
				\begin{bmatrix}
					L_2 & 0 & -L_m & 0\\
					0 & L_2 & 0 & -L_m\\
					-L_m & 0 & L_1 & 0\\
					0 & -L_m & 0 & L_1
				\end{bmatrix},
		\end{equation}
		where 
		\begin{equation}
			\begin{split}
				\bar{L}&\coloneqq L_1L_2-L_m^2,\qquad\alpha \coloneqq \bar{L}\omega_0, \qquad \beta\coloneqq L_m^2\omega_r\\
				\beta_{12} &\coloneqq L_1L_2\omega_r, \qquad \beta_1 \coloneqq L_1L_m\omega_r,\qquad\beta_2 \coloneqq L_2L_m\omega_r.
			\end{split}
		\end{equation}
		Parameters $R_1, R_2$ are the stator and rotor resistances, while $L_1, L_2, L_m$ are the stator and rotor auto-inductances and the mutual inductance, respectively.
		Finally, $\omega_r$ and $\omega_0$ are the electrical angular speeds of the rotor and the rotating reference frame, which we suppose constant. 
		\begin{remark}
			We suppose to have uncertainties on the parameters $R_1$ and $R_2$.
			This makes the matrix $A$ uncertain in half of its entries.
			In this example $B$ is such that $\Image(\mapB)=\R^{\dimx\times\dimx}$, so Assumption \ref{as:matching} is fulfilled for any $A_0\in \R^{\dimx\times\dimx}$.
		\end{remark}
		Denote the true resistances as $R_1, R_2$.
		We model our uncertainties specifying nominal values $\bar{R}_1, \bar{R}_2$ and radiuses $r_1, r_2>0$ such that 
		\begin{equation}\label{eq:resistance_bounds}
			\begin{split} 
				R_1&\in[\bar{R}_1-r_1, \bar{R}_1+r_1]\\
				R_2&\in[\bar{R}_2-r_2, \bar{R}_2+r_2].
			\end{split}
		\end{equation}
		Next, we define $\mathcal{C}$ as a ball about the nominal $\bar{A}$ (i.e., having the structure \eqref{eq:motor_matrices} with resistances $\bar{R}_1$ and $\bar{R}_2$) containing all possible parameter variation, i.e., 
		\begin{equation}
			\mathcal{C}\coloneqq \{ \hA \in \R^{n\times n}: |\hA-\bar{A}|_F \leq \rho \}
		\end{equation}
		with $\rho>0$ big enough.
		We report in Table \ref{table:params} the physical parameters of the motor. In Table \ref{table:reference_uncertainties}, we specify the values used for the uncertainties and the desired performances.
		\addtocounter{table}{1}
		\begin{table}[t!]
			\caption{Physical parameters of the motor.}
			\label{table:params}
			\centering
			\begin{tabular}{c | c | c | c}
				Parameter & Value & Parameter & Value \\
				\hline
				$L_1$ [H] & $0.02645$ & $R_1$ [$\Omega$]& $0.036$ \\
				$L_2$ [H] & $0.0264$& $R_2$ [$\Omega$]& $0.038$ \\
				$L_m$ [H]& $0.0257$& $\omega_0$ [rad/s]& $2\pi70.8$\\
				$p$ & $3$ & $\omega_r$ [rad/s]& $2\pi62$  \\
			\end{tabular}
		\end{table}
		\begin{table}[t!]
			\caption{Uncertainty parameters for example 1.}
			\label{table:reference_uncertainties}
			\centering
			\begin{tabular}{c | c | c | c}
				Parameter & Value & Parameter & Value \\
				\hline
				$\bar{R}_1$  [$\Omega$]& $0.03$ & $r_1$  [$\Omega$]& $0.01$\\ 
				$\bar{R}_2$  [$\Omega$]& $0.03$ & $r_2$  [$\Omega$]& $0.01$\\
				$\rho$ & $20$ & &
			\end{tabular}
		\end{table}
		\begin{figure}[t!]
			\centering
			\includegraphics[width=0.7\linewidth]{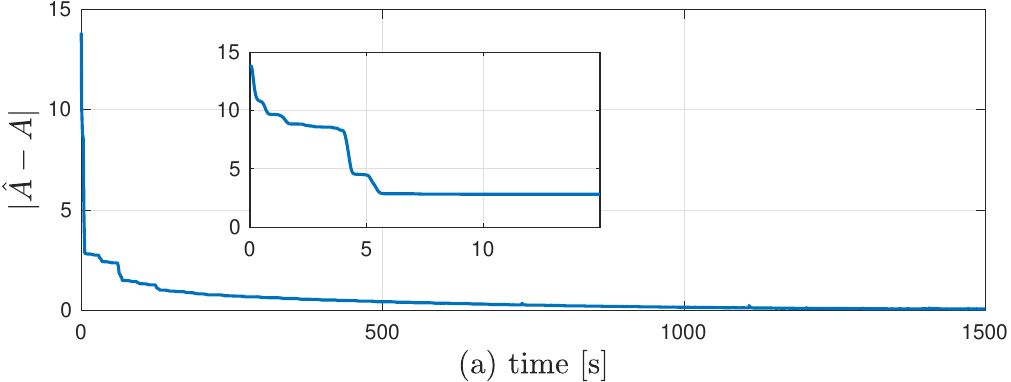}
			
			\includegraphics[width=0.7\linewidth]{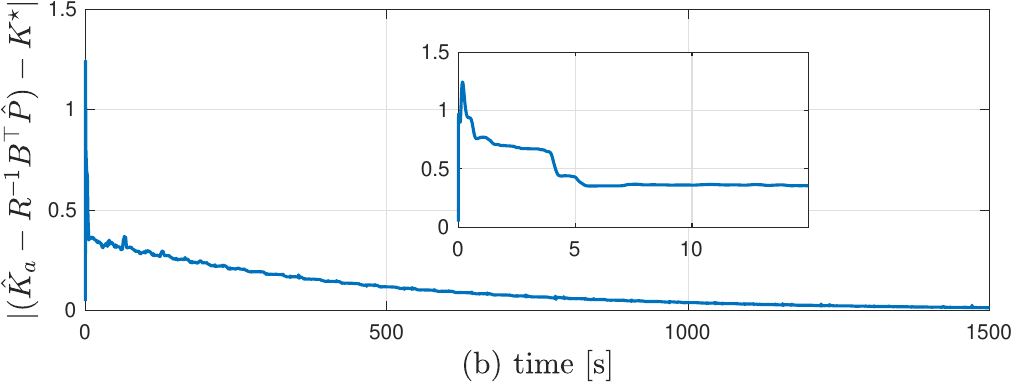}
			
			\caption{Convergence to true $A$ and to optimal gain $K^\star$.}
			\label{fig:estimations_convergence}
			
		\end{figure}
		\begin{figure}[t!]
			\centering
			\includegraphics[width=0.7\linewidth]{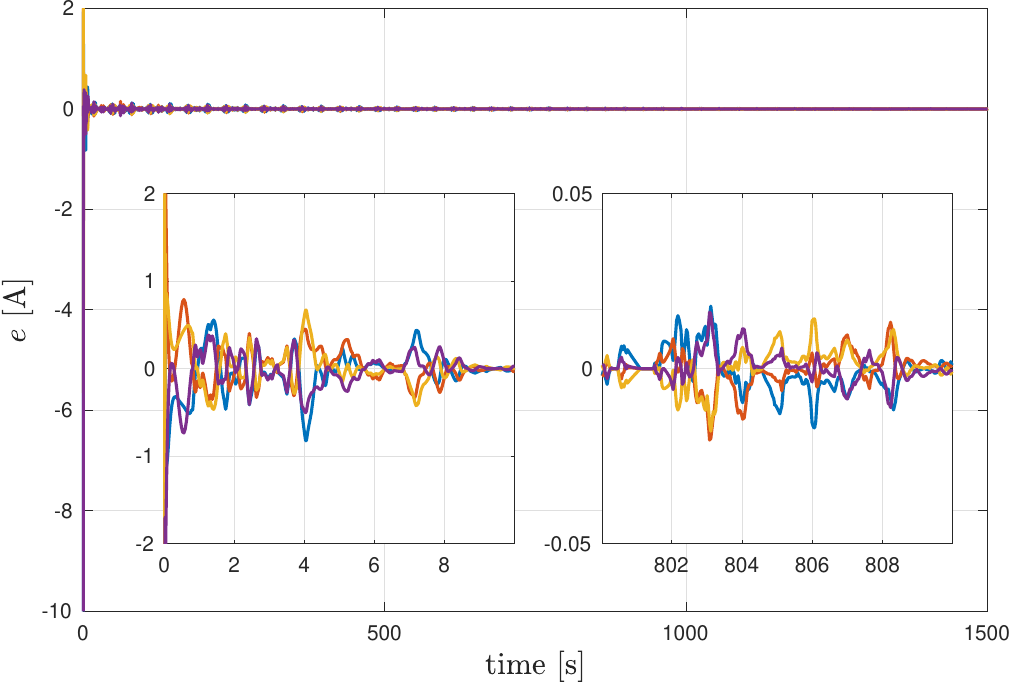}
			\caption{Tracking error between plant and reference model. Different colors stand for different components of $e$.}
			\label{fig:error}
			\vspace{-15pt}
		\end{figure}
		The dither $d(t)$ is designed, on each entry $d_i(t)$, according to
		\begin{equation}\label{eq:dither}
			d_i(t) = 10\sum_{j=1}^4 \text{sawtooth}(2\omega_s ij t), \quad i \in \{1, 2, 3, 4\}
		\end{equation}
		where $\text{sawtooth}(\cdot)$ is a triangular wave of unitary amplitude and $\omega_s=0.2$ rad/s. 
		In Fig. \ref{fig:estimations_convergence}-(a), we show the difference between the estimate $\hA(t)$ and the true matrix $A$. Next, in Fig. \ref{fig:estimations_convergence}-(b), we show how the error between the optimal feedback gain $K^\star$ and the overall applied feedback gain $-R^{-1}B^\top \hP + \hKe$ approaches zero, thus controlling in an optimal way the system. 
		In Fig. \ref{fig:error}, we show for completeness the error between the reference model and the real system, which reaches a small amplitude in a few seconds. 	
	
	\subsection{Example 2: Drifting Parameters and Variable Speed}
	
		In this example, we apply perturbations to the DFIM with model given in \eqref{eq:motor_matrices} to test the robustness of MR-ARL. We consider two perturbations to the nominal model occurring together: the first one is a time-varying resistance due to motor heating up, while the second one is a time-varying rotor speed due to load changes. 
		We model both disturbances with sigmoid functions and we report them in the plots. The temperature disturbance lasts for about $600$ s and brings the temperature from $20$ $^\circ$C to $100$ $^\circ$C, i.e.,  $\Delta T= 80$ $^\circ$C. The speed disturbance is a total increase of speed of $2\pi20$~rad/s occurring in about $60$ s.
		We model the dependence of resistances on temperature with
		\begin{equation}
			R_i (\Delta T) = R_i + \alpha\Delta T,\qquad i \in\{1, 2\},
		\end{equation}
		where $\alpha_{\text{CU}}=4.041\times 10^{-3}$ $\Omega/^\circ$C is the temperature coefficient of resistance of the copper. 
		
		We set new nominal $\bar{R}_1, \bar{R}_2, \bar{\omega}_r$ with associated range $r_1, r_2, r_\omega$ (reported in Table \ref{table:new_uncertainties}) to consider these uncertainties.
		We recalculate $\mathcal{C}$ as in the previous example.  
		Finally, we leave the dither as in \eqref{eq:dither}. 
		
		\begin{table}
			\caption{Uncertainty parameters for example 2.}
			\label{table:new_uncertainties}
			\centering
			\begin{tabular}{c | c | c | c}
				Parameter & Value & Parameter & Value \\
				\hline
				$\bar{R}_1$  [$\Omega$]& $0.2$ & $r_1$  [$\Omega$]& $0.18$ \\
				$\bar{R}_2$  [$\Omega$]& $0.2$ & $r_2$  [$\Omega$]& $0.18$ \\
				$\bar{\omega}_r$  [rad/s]& $2\pi70$ & $r_\omega$  [rad/s]& $2\pi15$ \\
				$\rho$ & $4830$ & &
			\end{tabular}
		\end{table}
		
		\begin{remark}
			Due to the parameter variations, the plant becomes a slowly time-varying system.
			Consistently with the theoretical result, due to the ``small'' variations, the stability properties of Theorem~\ref{theorem_DRE} are practically preserved and recovered when the variations vanish.
		\end{remark}
		
		In Fig. \ref{fig:estimations_convergence2}-(a), we show the difference between the estimate $\hA(t)$ and the true time-varying matrix $A(t)$.
		Notice that as soon as the speed disturbance ends, the gradient estimator is able to adapt and recover convergence of the estimation to a small ball about the true parameters. 
		Next, in Fig. \ref{fig:estimations_convergence2}-(b), we show how the data-driven feedback gain approaches the optimal one.
		Since in this simulation we have a LTV plant, we calculate at each time instant the optimal gain $K^\star (t)$ by solving an LQR problem with constant $A(t)$.
		The importance of the adaptive controller action is particularly clear in presence of the speed disturbance, where the estimated matrix is far from the true one and thus the optimal action is likely to be destabilizing. 
		
		Finally, we show in Fig. \ref{fig:e_2} how the error between the reference model and the real plant is kept bounded also in the presence of these disturbances. 
		
		\begin{figure}
			\centering
			\hspace{-12pt}
			\includegraphics[width=0.7\linewidth]{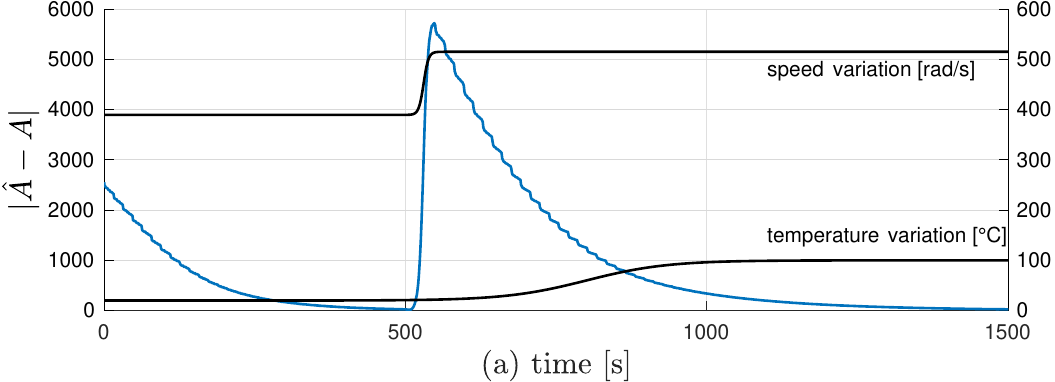}
			
			\includegraphics[width=0.7\linewidth]{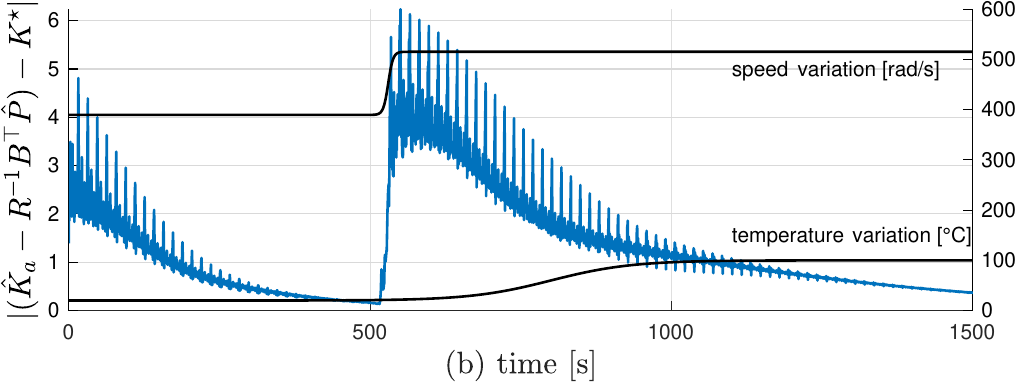}
			
			\caption{Convergence to true $A(t)$ and to optimal gain $K^\star(t)$.}
			\label{fig:estimations_convergence2}
			
		\end{figure}
		
		\begin{figure}
			\centering
			\includegraphics[width=0.7\linewidth]{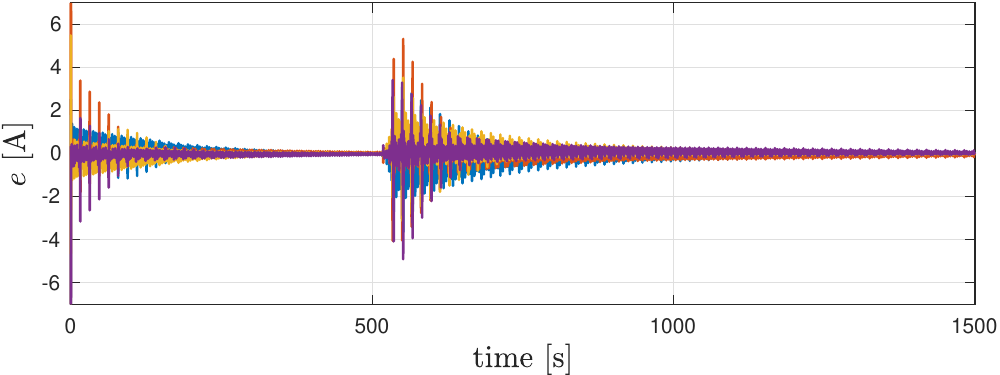}
			
			\caption{Tracking error between plant and reference model. Different colors stand for different components of $e$.}
			\label{fig:e_2}
			
		\end{figure}
		
\section{Conclusions}\label{sec:conclusion}
	
	In this paper, we have addressed the problem of data-driven optimal control of partially unknown linear systems.
	First, we have proposed a framework that formalizes a robustlty stable on-policy data-driven LQR problem in which optimality of the learned strategy is obtained while guaranteeing robust stability of the whole learning and control closed-loop system.
	Next, we have proposed a new solution to this problem consisting in the combination of model reference adaptive control and reinforcement learning.
	As main result, we showed that our design has a semiglobally uniformly asymptotically stable attractor where the plant follows the optimal reference model.
	To demonstrate the effectiveness of the solution, we tested it in the control of a doubly fed induction motor.
	The results show that our solution is also able to manage non-vanishing perturbations typical of real-world applications.
	
\section{Appendix}

	\subsection{Proofs}
	\begin{proof}[Proof of Lemma \ref{lemma:bounded_id}]
		At first, to simplify the expressions for the Lyapunov function, we introduce the following vectorized coordinates:
		\begin{equation}
			\begin{split}
				\hTa \coloneqq \vecc(\hA) \in \R^{n^2},\qquad\tilde{\theta}_\text{A} \coloneqq \vecc(\tilde{A}) \in \R^{n^2}.
			\end{split}
		\end{equation}
		It can be verified the following relation holds:
		\begin{equation}\label{eq:vectorized_gradient_before_projection}
			\begin{split}
				\vecc(BB^\dagger \epsilon\xi^\top ) =  \bar{B}(\xi \otimes I_n)\epsilon,\;\;\;\;\; \bar{B}\coloneqq (I_n \otimes BB^\dagger),
			\end{split}
		\end{equation}
		where $\bar{B}$ defines a projection onto $\Image(I_n \otimes B)\subset \R^{n^2}$.
		Notice that, for any $\hTa \in \Image(I_n \otimes B)$ and $\tau \in \R^{n^2}$, then $\tilde{\theta}_\text{A} \in \Image(I_n \otimes B)$ and, since the scalar product of orthogonal vectors is zero and by idempotence of the projection,
		\begin{equation}\label{eq:projection_relation}
				\tilde{\theta}_\text{A}^\top \tau = \tilde{\theta}_\text{A}^\top (\tau_\parallel + \tau_\perp) = \tilde{\theta}_\text{A}^\top \tau_\parallel + 0 = \tilde{\theta}_\text{A}^\top \bar{B}\tau_\parallel=\tilde{\theta}_\text{A}^\top \bar{B}(\tau_\parallel+\tau_\perp)= \tilde{\theta}_\text{A}^\top \bar{B}\tau,
		\end{equation}
		where $\tau_\parallel \in \Image(I_n \otimes B)$ and $\tau_\perp \in (I_n \otimes B)^\perp$.
		We rewrite \eqref{eq:identifier} by using the vectorized coordinates defined above:	
		\begin{equation}\label{eq:identifier_vectorized}
			\begin{split}
				\dhTa = \proj_{\vecc^{-1}\hTa \in \mathcal{C}} \!
				\left\{- \gamma\bar{B} \frac{(\xi \otimes I_{\dimx})\epsilon}{1 + \nu |\xi||\epsilon|}\right\}.
			\end{split}
		\end{equation}
		The computations from here are similar to \cite[Lemma 6.1]{krstic1995nonlinear} but we report them for the reader's convenience.
		Define
		\begin{equation}\label{eq:identifier_Lyapunov}
			V_A(\tilde{\epsilon}, \tTa) \coloneqq  \frac{1}{\lambda}|\tilde{\epsilon}|^2+ \frac{1}{2\gamma}|\tTa|^2 ,
		\end{equation}
		which is positive definite with respect to $(0, 0)$ and radially unbounded.
		Note that $\epsilon = (\xi \otimes I_n)^\top \tTa + \tilde{\epsilon}$.
		Then, using \eqref{eq:projection_relation} and \cite[Lemma E.1]{krstic1995nonlinear} to treat the projection operator $\proj_{\hA \in \mathcal{C}}\{\cdot\}$, the derivative of $V_A$ along the solutions of \eqref{eq:id_error_dyn} is
		\begin{equation}
			\begin{split}
				\dot{V}_A &= - 2|\tilde{\epsilon}|^2 + \tTa^\top\proj_{\vecc^{-1}\hTa \in \mathcal{C}} \!
				\left\{- \bar{B} \frac{(\xi \otimes I_{\dimx})\epsilon}{1 + \nu |\xi||\epsilon|}\right\} \\
				&\leq - 2|\tilde{\epsilon}|^2 - \frac{\tTa^\top \bar{B} (\xi \otimes I_{\dimx})\epsilon}{1 + \nu |\xi||\epsilon|}  = - 2|\tilde{\epsilon}|^2 - \frac{\tTa^\top (\xi \otimes I_{\dimx})\epsilon}{1 + \nu |\xi||\epsilon|} \\
				&\leq - 2|\tilde{\epsilon}|^2- \frac{|\epsilon|^2 - \tilde{\epsilon}^\top\epsilon}{1 + \nu |\xi|^2} \\
				&\leq - 2|\tilde{\epsilon}|^2- \frac{1}{4}\frac{|\epsilon|^2}{(1 + \nu |\xi|^2)^2} + \frac{\tilde{\epsilon}^\top\epsilon}{1 + \nu |\xi|^2}-\frac{3}{4}\frac{|\epsilon|^2}{1 + \nu |\xi|^2}  \\
				&= - |\tilde{\epsilon}|^2 - \left(\frac{1}{2}\frac{|\epsilon|}{1 + \nu |\xi|^2} - \tilde{\epsilon} \right)^2 -\frac{3}{4}\frac{|\epsilon|^2}{1 + \nu |\xi|^2} \leq 0\\
			\end{split}
		\end{equation}
		implying that $(\tilde{\epsilon}(t), \tilde{\theta}_A(t))$ is contained for all $t \in [0, t_f)$ in a compact sublevel set of $V_A$.
		We conclude the proof by recalling \cite[Lemma E.1]{krstic1995nonlinear} to ensure $\hA(t) \in \KAL$, for all its domain of existence.
	\end{proof}
	\begin{proof}[Proof of Lemma \ref{lemma:bounded_A_dot}]
		Using \cite[Lemma E.1]{krstic1995nonlinear} to treat the projection operator $\proj_{\hA \in \mathcal{C}}\{\cdot\}$ and the fact that $|\bar{B}| = 1$ due to \eqref{eq:vectorized_gradient_before_projection}, we can bound $|\dhA|$ as follows:
		\begin{equation}\label{eq:bound_for_hA_dot}
			|\dhA| \leq |\dhA|_F = |\dhTa| \leq \gamma|\bar{B}|\frac{|(\xi \otimes I_n)||\epsilon|}{1+\nu|\xi||\epsilon|} \leq \frac{\gamma|\xi||\epsilon|}{1+\nu|\xi||\epsilon|} \leq \gamma.
		\end{equation}
	\end{proof}
	
	\begin{proof}[Proof of Lemma \ref{lemma:MR_ISS}]
		Function $\Pta$ being continuous and $\KAL$ a compact set, there exist scalars $p_{\text{min}}$, $p_{\text{max}} > 0$ such that
		\begin{equation}\label{eq:P_bounds}
			p_{\text{min}}I_n\leq\Pta \leq p_{\text{max}}I_n, \quad \forall \hA \in \KAL.
		\end{equation}
		Then, define the Lyapunov function
		\begin{equation}\label{eq:Lyapunov_for_ref_model}
			V_{\text{m}}(\xm, t)\coloneqq \xm^\top \Pta \xm 
		\end{equation}
		which is positive definite and radially unbounded, and whose derivative along the solutions of \eqref{eq:model_reference} is given by:
		\begin{equation}\label{eq:vm_dot}
			\begin{split}
				\dot{V}_{\text{m}}=\; & \xm^\top \left( \Pta \Acl + \! \Acl^\top\!\! \Pta \right)\xm  + \xm^\top \left(\frac{\partial\Pta}{\partial \hA}\odot\dhA \right)\xm +2\xm^\top \Pta Bd,
			\end{split}
		\end{equation}
		where 
		\begin{equation}\label{eq:Acl_def}
			\Acl \coloneqq \hA - BR^{-1}B^\top \Pta
		\end{equation}
		and the product $\odot$ is defined as
		\begin{equation}
			\begin{split}
				\frac{\partial\Pta}{\partial \hA}\odot\dhA = \sum_{i=1}^n\sum_{j=1}^n \frac{\partial\Pta}{\partial [\hA]_{ij}}[\dhA]_{ij},
			\end{split}
		\end{equation}
		with $[\hA]_{ij}$ the $i$-th row and $j$-th column entry of matrix $\hA$.
		Since $\hP = \Pta$ solves at each time instant ARE~\eqref{eq:ARE_implementation}, it holds that:
		\begin{equation}\label{eq:are_lyap}
			\begin{split}
				\Acl&\Pta \! + \! \Pta \Acl=\!\underbrace{-Q\!-\!\Kta^\top \!\!R\Kta}_{\eqqcolon -\bar{Q}(\hA)},
			\end{split}
		\end{equation}
		where from Assumption \ref{as:ctrl_det} and $\KAL$ being compact, $\bar{Q}(\hA) \geq q>0$ for all $\hA \in \KAL$, with $q$ defined as
		\begin{equation}\label{eq:q_bound}
			q \coloneqq \min_{\hA \in \KAL} \lambda_{\min}\left(-Q-\Pta B RB^\top\Pta\right),
		\end{equation}
		where $\lambda_{\min}(\cdot)$ denotes the smallest eigenvalue of a matrix.
		Define $c\coloneqq \max_{i,j\in\{1, \ldots, n\}} \left\{ \max_{\hA \in \KAL} \left| \frac{\partial\Pta}{\partial[\hA]_{ij}}\right|\right\}$,
		then we obtain
		\begin{equation}\label{eq:P_derivative_bound}
			\begin{split}
				\left|\frac{\partial\Pta}{\partial \hA}\odot\dhA\right| &= \left|\sum_{i=1}^n\sum_{j=1}^n \frac{\partial\Pta}{\partial [\hA]_{ij}}[\dhA]_{ij}\right|\\
				&\leq c \sum_{i=1}^n\left(\sum_{j=1}^n |[\dhA]_{ij}|\right) \leq cn\max_{1 \leq i \leq n}\sum_{j=1}^n |[\dhA]_{ij}| \leq c n^{\frac{3}{2}} |\dhA|.
			\end{split}
		\end{equation}
		By letting $|\dhA| \leq \gamma^\star_b \coloneqq q/(2 c n^{\frac{3}{2}})$, \eqref{eq:vm_dot} becomes
		\begin{equation}
			\begin{split}
				\dot{V}_{\text{m}} = &\; -\xm^\top \left(\bar{Q}(\hA) - \frac{\partial \Pta}{\partial \hA}\odot\dhA  \right)\xm +  2\xm^\top \Pta Bd\\
				\leq & -(q - c\rho|\dhA| ) |\xm|^2 +  2 p_{\max}|B||\xm||d|\\
				\leq & -\frac{q}{2}|\xm|\left(|\xm| - \frac{4p_{\max}|B| |d|}{q}\right).
			\end{split}
		\end{equation}
		Therefore,
		\begin{equation}\label{eq:ultimate_bound_ref_model}
			|\xm| \geq \frac{8p_{\max}|B| |d|}{q} \implies \dot{V}_{\text{m}} \leq -\frac{q}{4}|\xm|^2,
		\end{equation}
		which concludes the statement.
	\end{proof}
	
	\begin{proof}[Proof of Lemma \ref{lemma:AC_ugs}]
		At first, to simplify the expressions for the Lyapunov function, we introduce the following vectorized coordinates:
		\begin{equation}
			\hTe \coloneqq \vecc(\hKe) \in \R^{mn}, \qquad \tTe  \coloneqq \vecc(\tKe)\in \R^{mn}.
		\end{equation} 
		We rewrite dynamics \eqref{eq:error_dyn} and \eqref{eq:adaptive_error_system} by using the vectorized coordinates above defined:
		\begin{equation}\label{eq:vectorized_adaptive_error_system}
			\begin{split}
				\de &=  \Acl e + B\hKe x - B\Ke x\\
				&=\Acl e + B(x \otimes I_{\dimu})^\top \tTe,\\
				\dtTe &= -\mu (x \otimes I_{\dimu}) B^\top\Pta e.
			\end{split}
		\end{equation}
		Consider the Lyapunov function
		\begin{equation}
			\begin{split}
				V_e(e, \tTe, t):= e^\top \Pta e + \frac{1}{\mu} | \tTe|^2,
			\end{split}
		\end{equation}
		which is positive definite and radially unbounded.
		The time derivative of $V_e$ along the trajectories of~\eqref{eq:vectorized_adaptive_error_system} is given by
		\begin{equation}\label{eq:daptive_error_dyn}
			\begin{split}
				\dot{V}_e =&\; e^\top \left(\Pta \Acl+ \Acl^\top \Pta + \frac{\partial \Pta}{\partial \hA}\odot\dhA\right)e\\
				&+2e^\top\Pta B(x \otimes I_{\dimu})^\top\tTe -\frac{2}{\mu} \tTe^\top (\mu (x \otimes I_{\dimu}) B^\top\Pta e)\\
				=& -e^\top \left(\bar{Q}(\hA) - \frac{\partial\Pta}{\partial \hA}\odot \dhA \right)e \leq -\frac{q}{2} |e|^2 \leq 0,
			\end{split}
		\end{equation}
		where $\Acl$ is defined in \eqref{eq:Acl_def}, $\bar{Q}(\hA)$ is given in~\eqref{eq:are_lyap}, and $q$ is found in~\eqref{eq:q_bound}.
		We have ensured that $(e(t), \tTe(t))$ is contained for all $t \in [0, t_f)$ in a compact sublevel set of $V_e$, thus concluding the proof.
	\end{proof}
	
	\begin{proof}[Proof of Lemma \ref{lemma:x_m_uPE}]
		For all $\hA \in \KAL$, pair $(\Acl, B)$ in \eqref{eq:model_ref_detailed}, with $\Acl$ given in \eqref{eq:Acl_def}, is controllable because $(\hA, B)$ is controllable from Assumption \eqref{as:ctrl_det}.
		Additionally, the origin of system $\dxm = \Acl\xm$ is UGES from Lemma \ref{lemma:MR_ISS}.
		\par From classical results on PE \cite[\S 5.6.4]{ioannou1996robust}, if $\dhA = 0$ then $\xm(t)$ is PE, i.e, there exist $T > 0$, $\alpha > 0$ such that
		\begin{equation}
			\int_{t}^{t + T}\xm(s)\xm(s)^\top \text{d}s \geq \alpha I_{\dimx}, \quad \forall t \geq 0.
		\end{equation}
		By \cite[Thm. 6.1]{mareels1988persistency}, there exist a constant scalar $\eta > 0$, such that, if
		\begin{equation}\label{eq:A_cl_slowly_varying}
			|A_{\text{cl}}(\hA(s)) - A_{\text{cl}}(\hA(\tau))| \leq \eta, \quad \forall s, \tau \in [t, t + T],
		\end{equation}
		for all $t\geq0$, then $\xm(t)$ is PE also when $\dhA \neq 0$.
		Recall that $\Acl \coloneqq \hA - BR^{-1}B^\top\Pta$ is an analytic function of $\hA$ \cite[Thm. 4.1]{ran1988parameter}. Thus, from the mean-value theorem and similar computations to \eqref{eq:P_derivative_bound}, we obtain:
		\begin{equation}
			\begin{split}
				|A_{\text{cl}}(\hA(s)) - A_{\text{cl}}(\hA(\tau))| &= \left|(s-\tau)\frac{\partial \Acl}{\partial \hA}\odot\dhA(\varsigma)\right|\\
				&\leq |s-\tau|c n^{\frac{3}{2}} |\dhA(\varsigma)|,
			\end{split}
		\end{equation}
		where $\varsigma \in [s, \tau]$, $c\coloneqq \max_{i,j \in \{1, \ldots, n\}} \left\{\max_{\hA \in \KAL}\left|\frac{\partial \Acl}{\partial [\hA]_{ij}}\right|\right\}$.
		From the fact that $|\dhA(\cdot)|\leq \gamma$, we conclude that for  $\gamma^\star_{PE}\coloneqq \eta/(Tcn^{\frac{3}{2}})$, if $\gamma \in (0, \gamma^\star_{PE}]$, then bound $\eta$ in \eqref{eq:A_cl_slowly_varying} is enforced and thus $\xm(t)$ is PE.
	\end{proof}
	
	\begin{proof}[Proof of Lemma \ref{lemma:AC_UGAS}]
		From Lemma \ref{lemma:AC_ugs} and the solutions being forward complete, it holds that the origin $(e, \tTe)=0$ of system \eqref{eq:vectorized_adaptive_error_system} is UGS.
		Note that the regressor in \eqref{eq:adaptive_error_system} is given by $x(t)$.
		Therefore, if $x(t)$ is uniformly PE (u-PE) as in \cite[Def. 5]{panteley2001relaxed}, 
		then UGAS and ULES of $(e, \tTe) = 0$ follows from \cite[Thm. 1 and 2]{panteley2001relaxed}.
		To prove u-PE of $x(t)$, note that $x(t) = \xm(t) + e(t)$, where $\xm(t)$ is PE from Lemma \ref{lemma:x_m_uPE}.
		Therefore, we conclude u-PE of $x(t)$ from \cite[Prop. 2]{panteley2001relaxed}.
	\end{proof}
	\vspace{0.5cm}
	\begin{proof}[Proof of Lemma \ref{lemma:id_UGES}]
		From Lemma \ref{lemma:bounded_id} and UGES of the $\tilde{\epsilon}$ subsystem, we only need to prove UGES of system \eqref{eq:id_error_dyn} with $\tilde{\epsilon} = 0$, which we write here in vectorized coordinates:
		\begin{equation}\label{id_reduced_dyn}
			\dot{\tilde{\theta}}_A =\! \proj_{\vecc^{-1}\hTa \in \mathcal{C}} \!
			\left\{- \gamma \bar{B} \frac{(\xi \otimes I_n)(\xi \otimes I_n)^\top\tilde{\theta}_A}{1 + \nu |\xi||\tA \xi|}\right\}.
		\end{equation}
		Since the directions where learning happens are unchanged by the projection operator and by $\bar{B}$, we are interested in studying regressor $\bar{\xi}(t)\coloneqq\frac{\xi(t)\otimes I_n}{\sqrt{1+\nu|\xi(t)||\tA(t)\xi(t)|}}$ in order to prove our result.
		Given a small enough gain $\gamma$, it holds from Lemma~\ref{lemma:x_m_uPE} that $x_m(t)$ is PE, while $e(t) \to 0$ exponentially fast from Lemma \ref{lemma:AC_UGAS}.
		From \eqref{eq:xi_dyn}, $\xi(t)$ is a filtered version of the PE signal $\xm(t) + e(t)$, thus $\xi(t)$ is PE \cite[Lemma. 4.8.3]{ioannou1996robust}.
		Since all signals are bounded and $(\xi \otimes I_n)(\xi \otimes I_n)^\top = (\xi\xi^\top)\otimes I_n$, PE of $\xi(t)$ implies that
		\begin{equation}
			\int_{t}^{t + T}\bar{\xi}(s)\bar{\xi}(s)^\top\textup{d}s \geq 	\int_{t}^{t + T}\!\! \frac{(\xi(s)\xi(s)^\top)\otimes I_n}{1+\xi_M^2\tA_M}\textup{d}s \geq \alpha I_{\dimx^2},
		\end{equation}
		for some $T$, $\alpha > 0$ and all $t \in \mathbb{R}_{\geq 0}$, with $\xi_M\coloneqq \sup_{t \in \R} |\xi(t)|$, $\tA_M\coloneqq \sup_{t \in \R} |\tA(t)|$, thus $\bar{\xi}(t)$ is PE.
		From \cite[Thm. 8.5.6]{ioannou1996robust}, we conclude that $\tA = 0$ is UGES.
	\end{proof}
	\vspace{0.5cm}
		
	\bibliography{adaptive_rl_bib}
	\bibliographystyle{IEEEtran}

\vfill\null

\end{document}